%% file: main.tex
\documentclass[sigplan,9pt,review=false]{acmart}

\usepackage[utf8]{inputenc}
\usepackage[T1]{fontenc}
\input{glyphtounicode}

\acmYear{2018}
\acmConference[LICS '18]{LICS '18: 33rd Annual ACM/IEEE Symposium on
Logic in Computer Science}{July 9--12, 2018}{Oxford, United Kingdom}
\acmBooktitle{LICS '18: LICS '18: 33rd Annual ACM/IEEE Symposium on
Logic in Computer Science, July 9--12, 2018, Oxford, United Kingdom}
\acmPrice{15.00}
\acmDOI{10.1145/3209108.3209161}
\acmISBN{978-1-4503-5583-4/18/07}

\setcopyright{none}
\copyrightyear{2018}

\usepackage{subcaption} 
\usepackage{graphicx}
\usepackage{amssymb,amsmath,mathtools}
\usepackage{epstopdf}
\usepackage{url}
\usepackage{csquotes}
\usepackage{wasysym}
\usepackage{tabularx}
\usepackage{float}
\usepackage{placeins}
\usepackage{cancel}
\usepackage{hyperref}
\usepackage{multirow}
\usepackage{tikz}
\usetikzlibrary{automata,positioning}

\usepackage{thm-restate}
\usepackage{textcomp}
\usepackage{xspace}
\usepackage{xcolor}
\usepackage{booktabs}
\usepackage{colortbl}

\tikzset{
    state/.style={
		        rectangle,
            rounded corners,
            draw=black,
            minimum height=2em,
            minimum width=2em,
            align=center,
            }
}

\tikzset{
    statep/.style={
            circle,
            draw=black,
            minimum height=2em,
            minimum width=2em,
            align=center,
            }
}

\tikzstyle{acc}=[double]

\newcommand{\true}{{\ensuremath{\mathbf{tt}}}}
\newcommand{\false}{{\ensuremath{\mathbf{ff}}}}
\newcommand{\F}{{\ensuremath{\mathbf{F}}}}
\renewcommand{\G}{{\ensuremath{\mathbf{G}}}}
\newcommand{\X}{{\ensuremath{\mathbf{X}}}}
\renewcommand{\U}{{\ensuremath{\mathbf{U}}}}
\newcommand{\W}{{\ensuremath{\mathbf{W}}}}
\newcommand{\M}{{\ensuremath{\mathbf{M}}}}
\newcommand{\R}{{\ensuremath{\mathbf{R}}}}

\newcommand{\subf}{{\it sf}}

\newcommand{\sfmu}{{\ensuremath{\mathbb{\mu}}}}
\newcommand{\sfnu}{{\ensuremath{\mathbb{\nu}}}}
\newcommand{\setF}{\ensuremath{\mathcal{F}}}
\newcommand{\setG}{\ensuremath{\mathcal{G}}}
\newcommand{\setFG}{\ensuremath{\mathcal{F\hspace{-0.1em}G}}}
\newcommand{\setGF}{\ensuremath{\mathcal{G\hspace{-0.1em}F}\!}}

\newcommand{\aft}{\textit{af}}
\newcommand{\daft}{\textit{af}^{~\vee}}

\newcommand{\Reach}{\textit{Reach}}
\newcommand{\DReach}{\textit{Reach}^{\vee}}

\newcommand{\dnf}{\textit{dnf}}

\newcommand{\evalnu}[2]{{\ensuremath{#1[#2]_{\nu}}}}
\newcommand{\evalmu}[2]{{\ensuremath{#1[#2]_{\mu}}}}

\newcommand{\ltlfum}{\mu{LTL}}
\newcommand{\ltlgrw}{\nu{LTL}}

\newcommand{\fin}{\textit{fin}\,}
\renewcommand{\inf}{\textit{inf}\,}
\newcommand{\trans}[1]{\overset{#1}{\longrightarrow}}
\newcommand{\qee}{\hfill$\triangle$} % quod erat exemplandum

\begin{document}

%% Title information
\title[A Unified Translation of LTL into $\omega$-Automata]{One Theorem to Rule Them All:\\A Unified Translation of LTL into $\omega$-Automata}
\titlenote{This work was partially funded and supported by the Czech Science Foundation, grant No.~\mbox{P202/12/G061}, and the German Research Foundation (DFG) project \enquote{Verified Model Checkers} (317422601).}

%% Author information
%% Contents and number of authors suppressed with 'anonymous'.
%% Each author should be introduced by \author, followed by
%% \authornote (optional), \orcid (optional), \affiliation, and
%% \email.
%% An author may have multiple affiliations and/or emails; repeat the
%% appropriate command.
%% Many elements are not rendered, but should be provided for metadata
%% extraction tools.

\author{Javier Esparza}    
\email{esparza@in.tum.de}             

\affiliation{
  \institution{Technische Universität München} 
  \country{Germany}                 
}

\author{Jan K\v{r}et{\'i}nsk{\'y}}         
\email{jan.kretinsky@in.tum.de}
\orcid{0000-0002-8122-2881}           

\affiliation{
  \institution{Technische Universität München} 
  \country{Germany}                 
}

\author{Salomon Sickert}
\email{sickert@in.tum.de}               
\orcid{0000-0002-0280-8981}            

\affiliation{
  \institution{Technische Universität München} 
  \country{Germany}                 
}

\begin{abstract}
We present a unified translation of LTL formulas into deterministic Rabin automata, limit-deterministic Büchi automata, and nondeterministic Büchi automata. The translations yield automata of asymptotically optimal size (double or single exponential, respectively). All three translations are derived from one single Master Theorem of purely logical nature. The Master Theorem decomposes the language of a formula into a positive boolean combination of languages that can be translated into $\omega$-automata by elementary means. In particular, Safra's, ranking, and breakpoint constructions used in other translations are not needed. 
\end{abstract}

%% 2012 ACM Computing Classification System (CSS) concepts
%% Generate at 'http://dl.acm.org/ccs/ccs.cfm'.

%% End of generated code

%% Keywords
%% comma separated list

\maketitle

\input{Sections/introduction.tex}
\input{Sections/preliminaries.tex}
\input{Sections/construction-simple.tex}
\input{Sections/decomposition.tex}
\input{Sections/construction-dra.tex}

\input{Sections/construction-nba.tex}
\input{Sections/conclusion.tex}

\bibliography{ref}

\newpage
\appendix

\input{Appendix/proofs.tex}

\end{document}

%% file: Sections/introduction.tex
% !TEX root = ../main.tex

\newcommand{\para}[1]{\smallskip\noindent\textbf{#1}}

\section{Introduction}

Linear temporal logic (LTL) \cite{DBLP:conf/focs/Pnueli77} is a prominent specification language, used both for model checking and automatic synthesis of systems.
In the standard automata-theoretic approach \cite{DBLP:conf/lics/VardiW86} the input formula is first translated into an $\omega$-automaton, and then the product of this automaton with the input system is further analyzed.
Since the size of the product is often the bottleneck of all the verification algorithms, it is crucial that the $\omega$-automaton is as small as possible. 
Consequently, a lot of effort has been spent on translating LTL into small automata, e.g. \cite{DBLP:conf/fm/Couvreur99,DBLP:conf/cav/DanieleGV99,DBLP:conf/concur/EtessamiH00,DBLP:conf/cav/SomenziB00,DBLP:conf/cav/GastinO01,DBLP:conf/forte/GiannakopoulouL02,DBLP:conf/wia/Fritz03,DBLP:conf/tacas/BabiakKRS12,DBLP:conf/atva/Duret-LutzLFMRX16}.

While non-deterministic B\"uchi automata (NBA) can be used for model checking non-deterministic systems, other applications such as model checking probabilistic systems or synthesis usually require automata with a certain degree of determinism, such as deterministic parity automata (DPA) or deterministic Rabin automata (DRA) \cite{DBLP:books/daglib/0020348}, deterministic generalized Rabin automata (DGRA) \cite{DBLP:conf/cav/ChatterjeeGK13}, limit-deterministic (or semi-deterministic) B\"uchi automata (LDBA) \cite{DBLP:conf/focs/Vardi85,DBLP:journals/jacm/CourcoubetisY95,DBLP:conf/concur/HahnLST015,DBLP:conf/cav/SickertEJK16}, unambiguous B\"uchi automata \cite{DBLP:conf/cav/BaierK0K0W16} etc.
The usual constructions that produce such automata are based on Safra's determinization and its variants \cite{DBLP:conf/focs/Safra88,DBLP:conf/lics/Piterman06,DBLP:conf/fossacs/Schewe09}.
However, they are known to be difficult to implement efficiently, and to be practically inefficient in many cases due to their generality.
Therefore, a recent line of work shows how DPA \cite{DBLP:conf/tacas/EsparzaKRS17,DBLP:conf/tacas/KretinskyMWW17}, DRA and DGRA \cite{DBLP:conf/cav/KretinskyE12,DBLP:conf/atva/KretinskyL13,DBLP:conf/cav/EsparzaK14,DBLP:journals/fmsd/EsparzaKS16}, or LDBA \cite{DBLP:conf/tacas/KiniV15,DBLP:conf/cav/SickertEJK16,DBLP:conf/tacas/Kini017} can be produced directly from LTL, without the intermediate step through a non-deterministic automaton. 
All these works share the principle of describing each state by a collection of formulas, as happens in the classical tableaux construction for translation of LTL into NBA. This makes the approach particularly apt for semantic-based state reductions, e.g., for merging states corresponding to equivalent formulas. These reductions cannot be applied to Safra-based constructions, where this semantic structure gets lost.

In this paper, we provide a unified view of translations of LTL into NBA, LDBA, and DRA enjoying the following properties, absent in former translations: 

\paragraph{Asymptotic Optimality.} 
D(G)RA are the most compact among the deterministic automata used in practice, in particular compared to DPA.
Previous translations to D(G)RA were either limited to fragments of LTL \cite{DBLP:conf/cav/KretinskyE12,DBLP:conf/atva/KretinskyL13,DBLP:conf/atva/BabiakBKS13}, or only shown to be triply exponential \cite{DBLP:conf/cav/EsparzaK14,DBLP:journals/fmsd/EsparzaKS16}. 
Here we provide constructions for all mentioned types of automata matching the optimal double exponential bound for DRA and LDBA, and the optimal single exponential bound for NBA.

\paragraph{Symmetry.}
The first translations \cite{DBLP:conf/cav/KretinskyE12,DBLP:conf/atva/KretinskyL13} used auxiliary automata to monitor each \emph{Future}- and \emph{Globally}-subformula.
While this approach worked for fragments of LTL, subsequent constructions for full LTL \cite{DBLP:conf/cav/EsparzaK14,DBLP:journals/fmsd/EsparzaKS16,DBLP:conf/cav/SickertEJK16} could not preserve the symmetric treatment. They only used auxiliary automata for $\G$-subformulas, at the price of more complex constructions. Our translation re-establishes the symmetry of the first constructions. It treats $\F$ and $\G$ equally (actually, and more generally, it treats each operator and its dual equally), which results into simpler auxiliary automata.

\paragraph{Independence of Syntax.}
Previous translations were quite sensitive to the operators used in the syntax of LTL.
In particular, the only greatest-fixed-point operator they allowed was  \emph{Globally}. Since formulas also had to be
in negation normal form, pre-processing of the input often led to unnecessarily large formulas. While our translations still requires negation normal form, 
it allows for direct treatment of \emph{Release}, \emph{Weak until}, and other operators.

\paragraph{Unified View.} 
Our translations rely on a novel \emph{Master Theorem}, which decomposes the language of a formula into a positive boolean combination of ``simple'' languages, in the sense that they are 
easy to translate into automata. 
This approach is arguably simpler than previous ones (it is certainly simpler than our previous papers \cite{DBLP:journals/fmsd/EsparzaKS16,DBLP:conf/cav/SickertEJK16}). 
Besides, it provides a unified treatment of DRA, NBA, and LDBA, differing only in the translations of the ``simple'' languages. The automaton for the formula is obtained from the
automata for the ``simple'' languages by means of standard operations for closure under union and intersection.  

\smallskip

On top of its theoretical advantages, our translation is comparable to previous DRA translations in practice, even without major optimizations. Summarizing, we think this paper finally achieves the goals formulated in \cite{DBLP:conf/cav/KretinskyE12}, where the first translation of this kind---valid only for what we would now call a small fragment of LTL---was presented. 

\paragraph{Structure of the Paper.} 
Section \ref{sec:prelim} contains preliminaries about LTL and $\omega$-automata. Section \ref{sec:after} introduces some definitions and results of \cite{DBLP:journals/fmsd/EsparzaKS16,DBLP:conf/cav/SickertEJK16}. Section \ref{sec:fragments} shows how to use these notions to translate four simple fragments of LTL into deterministic Büchi and coBüchi automata; these translations are later used as building blocks. Section \ref{sec:master} presents our main result, the Master Theorem. Sections \ref{sec:ltl2dra}, \ref{sec:ltl2nba}, and \ref{sec:ltl2ldba} apply the Master Theorem to derive translations of LTL into DRA, NBA, and LDBA, respectively. Section \ref{sec:discussion} compares the paper to related work and puts the obtained results into context. The appendix of the accompanying technical report \cite{arXiv} contains the few omitted proofs and further related material.

%% file: Sections/preliminaries.tex
% !TEX root = ../main.tex

\section{Preliminaries}
\label{sec:prelim}

\subsection{$\omega$-Languages and $\omega$-Automata}

Let $\Sigma$ be a finite alphabet. 
An $\omega$-word $w$ over $\Sigma$ is an infinite sequence of letters $w[0] w[1] w[2] \dots$. We denote the finite infix
$w[i]w[i+1]\cdots w[j - 1]$ by $w_{ij}$, and the infinite suffix $w[i] w[i+1] \dots$ by $w_{i}$. An $\omega$-language is a set of $\omega$-words.

For the sake of presentation, we introduce $\omega$-automata with accepting conditions defined on states. However, all results can be restated with accepting conditions defined on transitions, more in line with other recent papers and tools \cite{DBLP:conf/atva/Duret-LutzLFMRX16,DBLP:conf/atva/KomarkovaK14,DBLP:conf/cav/BabiakBDKKM0S15}.

Let $\Sigma$ be a finite alphabet. A \emph{nondeterministic pre-automaton} over $\Sigma$ is a tuple $\mathcal{P} = (Q, \Delta, Q_0)$ where $Q$ is a finite set of states, $\Delta \colon  Q \times \Sigma \rightarrow 2^Q$ is a transition function, and $Q_0$ is a set of initial states. A transition is a triple $(q, a, q')$ such that $q' \in \Delta(q, a)$. A pre-automaton $\mathcal{P}$ is deterministic if $Q_0$ is a singleton and $\Delta(q, a)$ is a singleton for every $q \in Q$ and $a \in \Sigma$.

A \emph{run} of $\mathcal{P}$ on an $\omega$-word $w$ is an infinite sequence of states $r = q_0q_1q_2\dots$ with $q_{i+1} \in \delta(q_i, w[i])$ for all $i$  and we denote by $\inf(r)$ the set of states occurring infinitely often in $r$.  An \emph{accepting condition} is an expression over the syntax
$\alpha :: = \inf(S) \mid \fin(S) \mid \alpha_1 \vee \alpha_2 \mid \alpha_1 \wedge \alpha_2$ with $S \subseteq Q$.
Accepting conditions are evaluated on runs and the evaluation relation $r \models \alpha$ is defined as follows:
\begin{center}
\begin{tabular}{rclcl}
$r$  &$ \models$ & $\inf(S)$  & \ if{}f  \ & $\inf(r) \cap S \neq \emptyset$ \\
$r$  & $\models$ & $\fin(S)$  & if{}f  & $\inf(r) \cap S = \emptyset$ \\
$r$  & $\models$ & $\alpha_1 \vee \alpha_2$  & if{}f & $r \models \alpha_1$ or $r \models \alpha_2$ \\ 
$r$  & $\models$ & $\alpha_1 \wedge \alpha_2$  & if{}f & $r \models \alpha_1$ and $r \models \alpha_2$ \\
\end{tabular}
\end{center}
An accepting condition $\alpha$ is a 
\begin{itemize}
\item Büchi condition if $\alpha = \inf(S)$ for some set $S$ of states.
\item coBüchi condition if $\alpha = \fin(S)$ for some set $S$ of states.
\item Rabin condition if $\alpha = \bigvee_{i=1}^k (\inf(I_i) \wedge \fin(F_i))$ for some $k\geq 1$ and some sets $I_1, F_1, \ldots, I_k, F_k$ of states.
\end{itemize}

An $\omega$-automaton over $\Sigma$ is a tuple $\mathcal{A} = (Q, \Delta, Q_0, \alpha)$ where 
$(Q, \Delta, Q_0)$ is a pre-automaton over $\Sigma$ and $\alpha$ is an accepting condition. 
A run $r$ of $\mathcal{A}$ is \emph{accepting} if $r \models \alpha$. A word $w$ is accepted by $\mathcal{A}$ if some run 
of $\mathcal{A}$ on $w$ is accepting.
An $\omega$-automaton is a Büchi (coBüchi, Rabin) automaton if its accepting condition is a Büchi (coBüchi, Rabin) condition.

\paragraph{Limit-Deterministic Büchi Automata.}

Intuitively, a NBA is limit-deterministic if it can be split into a non-deterministic component without accepting states, and a deterministic component. The automaton can only accept by ``jumping'' from the non-deterministic to the deterministic component, but after the jump
it must stay in the deterministic component forever. Formally, 
a NBA $\mathcal{B} = (Q, \Delta, Q_0, \alpha)$ is {\em limit-deterministic} (LDBA) if
$Q$ can be partitioned into two disjoint sets $Q = Q_\mathcal{N} \uplus Q_\mathcal{D}$, s.t.
\begin{enumerate}
  \item $\Delta(q, \nu) \subseteq Q_\mathcal{D}$ and $|\Delta(q, \nu)| = 1$ for every $q \in  Q_\mathcal{D}$, $\nu \in \Sigma$, and
  \item $S \subseteq Q_\mathcal{D}$ for all $S \in \alpha$.
\end{enumerate}

\subsection{Linear Temporal Logic}
We work with a syntax for LTL in which formulas are written in negation-normal form, i.e., negations only occur in front of atomic propositions. For every temporal operator we also include in the syntax its dual operator. On top of the next operator $\X$, which is self-dual, we introduce temporal operators  $\F$ (eventually), $\U$ (until), and $\W$ (weak until), and their duals $\G$ (always), $\R$ (release) and $\M$ (strong release). The syntax may look redundant but as we shall see it is essential to include $\W$ and $\M$ and very convenient to include $\F$ and $\G$.

\paragraph{Syntax and semantics of LTL.}
A formula of LTL in {\em negation normal form} over a set of atomic propositions ($Ap$) is given by the syntax: \vspace{-1em}
\begin{align*}
\varphi::= & \; \true \mid \false \mid a \mid \neg a \mid \varphi \wedge \varphi \mid \varphi\vee\varphi \mid \X\varphi \\
	      \mid  & \; \F\varphi \mid \G\varphi \mid \varphi\U\varphi  \mid \varphi\W\varphi \mid \varphi\M\varphi \mid \varphi\R\varphi
\end{align*}
\noindent where $a \in Ap$. 
We denote $\subf(\varphi)$ the set of subformulas of $\varphi$. A subformula $\psi$ of $\varphi$ is called {\em proper} if it is neither a conjunction nor a disjunction, i,e., if the root of its syntax tree 
is labelled by either $a$, $\neg a$, or a temporal operator. %$\F$, $\G$, $\U$, $\R$, $\M$, $\W$ or $\X$. 
The satisfaction relation $\models$ between $\omega$-words over the alphabet $2^{Ap}$ and formulas is inductively defined as follows:
\[\begin{array}[t]{lclclcl}
w \models \true \\
w \not\models \false \\
w \models a & \mbox{ if{}f } & a \in w[0] \\
w \models \neg a & \mbox{ if{}f } & a \not \in w[0] \\
w \models \varphi \wedge \psi & \mbox{ if{}f } & w \models \varphi \text{ and } w \models \psi\\
w \models \varphi \vee \psi & \mbox{ if{}f } & w \models \varphi \text{ or } w \models \psi\\
w \models \X \varphi & \mbox{ if{}f } & w_1 \models \varphi\\
w \models \F \varphi & \mbox{ if{}f } & \exists k. \, w_k \models \varphi\\
w \models \G \varphi & \mbox{ if{}f } & \forall k. \, w_k \models \varphi\\
w \models \varphi\U \psi & \mbox{ if{}f } & \exists k. \, w_k \models \psi ~\text{ and }~ \forall j < k. \, w_j \models \varphi \\
w \models \varphi\W \psi & \mbox{ if{}f } & w \models \G\varphi ~\text{ or }~ w \models \varphi\U \psi \\
w \models \varphi\M \psi & \mbox{ if{}f } & \exists k. \, w_k \models \varphi ~\text{ and }~ \forall j \leq k. \, w_j \models \psi \\
w \models \varphi\R \psi & \mbox{ if{}f } & w \models \G\psi ~\text{ or }~ w \models \varphi\M \psi \\
\end{array}\]
Two formulas are equivalent if they are satisfied by the same words. We also introduce the stronger notion of propositional equivalence:

\begin{definition}[Propositional Equivalence]
Given a formula $\varphi$, we assign to it a propositional formula $\varphi_P$ as follows: 
replace every maximal proper subformula $\psi$ by a propositional variable $x_\psi$. 
Two formulas $\varphi, \psi$ are {\em propositionally equivalent}, denoted $\varphi \equiv_P \psi$, 
if{}f $\varphi_P$ and $\psi_P$ are equivalent formulas of propositional logic. 
The set of all formulas propositionally equivalent to $\varphi$ is denoted by $[\varphi]_P$.
\end{definition}

\begin{example}
Let $\varphi = \X b \vee (\G(a \vee \X b) \wedge \X b)$ with $\psi_1 = \X b$ and $\psi_2 = \G(a \vee \X b)$. We have $\varphi_P = x_{\psi_1} \vee (x_{\psi_2} \wedge x_{\psi_1}) \equiv_P x_{\psi_1}$. Thus $\X b$ is propositionally equivalent to $\varphi$ and  $\X b \in [\varphi]_P$. \qee
\end{example}

\noindent Observe that propositional equivalence implies equivalence, but the converse does not hold.

\section{The \enquote{after} Function}
\label{sec:after}

We recall the definition of the\enquote{after function} $\aft(\varphi, w)$, read ``$\varphi$ after $w$'' \cite{DBLP:conf/cav/EsparzaK14,DBLP:journals/fmsd/EsparzaKS16}. The function assigns to a formula $\varphi$ and a finite word $w$ another formula such that, intuitively, $\varphi$ holds for $ww'$ if{}f $\aft(\varphi, w)$ holds ``after reading $w$'', that is,  if{}f $w' \models \aft(\varphi, w)$.\footnote{There is a conceptual correspondences to the derivatives of \cite{DBLP:journals/jacm/Brzozowski64} and $\aft$ directly connects to the classical \enquote{LTL expansion laws} \cite{DBLP:books/daglib/0020348}. Furthermore, the yet to be introduced $\daft$ relates to \cite{DBLP:journals/tcs/Antimirov96} in a similar way.}

\begin{definition}\label{def:af}
Let $\varphi$ be a formula and $\nu \in 2^{Ap}$ a single letter. The formula $\aft(\varphi,\nu)$ is inductively defined as follows:
\[
\arraycolsep=1.8pt
\begin{array}{ll}
\aft(a,\nu)&= \begin{cases} \true & \mbox{if $a \in \nu$} \\ \false & \mbox{if $a \notin \nu$}\end{cases} \\[0.3cm]
\aft(\neg a,\nu)&= \begin{cases} \false & \mbox{if $a \in \nu$} \\ \true & \mbox{if $a \notin \nu$}\end{cases}
\end{array}
\qquad
\begin{array}{ll}
\aft(\true,\nu)& = \true \\[0.1cm]
\aft(\false,\nu)& = \false \\[0.1cm]
\aft(\varphi\wedge\psi,\nu)& = \aft(\varphi,\nu)\wedge\aft(\psi,\nu) \\[0.1cm]
\aft(\varphi\vee\psi,\nu)& = \aft(\varphi,\nu)\vee\aft(\psi,\nu)
\end{array}\]
\[\arraycolsep=1.8pt
\begin{array}{ll}
\aft(\X\varphi,\nu)&= \varphi\\[0.1cm]
\aft(\F\varphi,\nu)&= \aft(\varphi,\nu)\vee \F\varphi \\[0.1cm]
\aft(\G\varphi,\nu)&= \aft(\varphi,\nu)\wedge \G\varphi \\[0.1cm]
\aft(\varphi\U\psi,\nu)&=\aft(\psi,\nu)\vee(\aft(\varphi,\nu)\wedge \varphi\U\psi) \\[0.1cm]
\aft(\varphi\W\psi,\nu)&=\aft(\psi,\nu)\vee(\aft(\varphi,\nu)\wedge \varphi\W\psi) \\[0.1cm]
\aft(\varphi\M\psi,\nu)&=\aft(\psi,\nu)\wedge(\aft(\varphi,\nu)\vee \varphi\M\psi) \\[0.1cm]
\aft(\varphi\R\psi,\nu)&=\aft(\psi,\nu)\wedge(\aft(\varphi,\nu)\vee \varphi\R\psi)
\end{array}\]
Furthermore, we generalize the definition to finite words by setting $\aft(\varphi, \epsilon) = \varphi$ and
$\aft(\varphi, \nu w) = \aft(\aft(\varphi,\nu),w)$ for every $\nu \in 2^{Ap}$ and every finite word $w$.
Finally, we define the set of formulas \emph{reachable from $\varphi$} as $\Reach(\varphi) = \{[\psi]_P \mid \exists w. ~ \psi = \aft(\varphi,w)\}$.
\end{definition}

\begin{example}
\label{ex:aft}
Let $\varphi = a \vee (b \; \U \; c) $. We then have $\aft(\varphi, \{a\}) \equiv_P \true$, $\aft(\varphi, \{b\}) \equiv_P (b \; \U \; c)$, $\aft(\varphi, \{c\}) \equiv_P \true$, and $\aft(\varphi, \emptyset) \equiv_P \false$. \qee
\end{example}

The following lemma states the main properties of $\aft$, which are easily proved by induction on the 
structure of $\varphi$. For convenience we include the short proof in the appendix of \cite{arXiv}.

\begin{restatable}{lemma}{lemAf}\cite{DBLP:journals/fmsd/EsparzaKS16}\label{lem:af:prop}
\begin{itemize}
\item[(1)] For every formula $\varphi$, finite word $w\in (2^{Ap})^*$, and infinite word $w'\in (2^{Ap})^\omega$:
$ww' \models \varphi$ if{}f $w' \models \aft(\varphi, w)$
\item[(2)] For every formula $\varphi$ and finite word $w\in (2^{Ap})^*$: $\aft(\varphi, w)$ is a positive boolean combination of proper subformulas of $\varphi$.
\item[(3)] For every formula $\varphi$: If $\varphi$ has $n$ proper subformulas, then $\Reach(\varphi)$ has at most size $2^{2^n}$.
\end{itemize}
\end{restatable}

It is easy to show by induction that $\varphi \equiv_P \psi$ 
implies $\aft(\varphi, w) \equiv_P \aft(\psi, w)$ for every 
finite word $w$. We extend $\aft \,$ to equivalence classes by defining 
$\aft([\varphi]_P, w) := [\aft(\varphi, w)]_P$. Sometimes we abuse language and identify a 
formula and its equivalence class. For example, we write ``the states of the automaton are
pairs of formulas'' instead of ``pairs of equivalence classes of formulas''.

%% file: Sections/construction-simple.tex
% !TEX root = ../main.tex

\section{Constructing DRAs for Fragments of LTL}
\label{sec:fragments}

We show that the  function $\aft$ can be used to construct deterministic Büchi and coBüchi automata for some fragments of $LTL$. The constructions are very simple. Later, in Sections \ref{sec:ltl2dra}, \ref{sec:ltl2nba}, and \ref{sec:ltl2ldba} we use these constructions as building blocks for the translation of general LTL formulas. The fragments are:
\begin{itemize}
\item The $\mu$-fragment $\ltlfum$ and the $\nu$-fragment $\ltlgrw$.\\
$\ltlfum$ is the fragment of LTL restricted to temporal operators $\F, \U, \M$, on top of Boolean connectives $(\wedge, \vee)$, literals $(a, \neg a)$, and the next operator $(\X)$. $\ltlgrw$ is defined analogously, but with the operators $\G, \W, \R$. In the literature $\ltlfum$ is also called syntactic co-safety and $\ltlgrw$ syntactic safety.
\item The fragments $\G\F(\ltlfum)$ and $\F\G(\ltlgrw)$.\\
These fragments contain the formulas of the form $\G\F \varphi$, where $\varphi \in \ltlfum$, and 
$\F\G \varphi$, where $\varphi \in \ltlgrw$.
\end{itemize}

The reason for the names $\ltlfum$ and $\ltlgrw$ is that $\F, \U, \M$  are least-fixed-point operators, in the sense that their semantics is naturally formulated by least fixed points, e.g. in the $\mu$-calculus, while the semantics of
$\G, \W, \R$ is naturally formulated by greatest fixed points. 

The following lemma characterizes the words $w$ satisfying a formula $\varphi$ of these fragments in terms of the formulas $\aft(\varphi, w)$.
\begin{lemma}\cite{DBLP:journals/fmsd/EsparzaKS16}\label{thm:logical:fragments}
Let $\varphi \in \ltlfum$ and let $w$ be a word. We have:
\begin{itemize}
\item $w \models \varphi$ if{}f $~\exists i. ~ \aft(\varphi, w_{0i}) \equiv_P \true$.
\item $w \models \G\F \varphi$ if{}f $~\forall i.~\exists j. ~ \aft(\F\varphi, w_{ij}) \equiv_P \true$.
\end{itemize}
Let $\varphi \in \ltlgrw$ and let $w$ be a word. We have:
\begin{itemize}
\item $w \models \varphi$ if{}f $~\forall i. ~ \aft(\varphi, w_{0i}) \not \equiv_P \false$.
\item $w \models \F\G\varphi$ if{}f $~\exists i.\forall j. ~ \aft(\G\varphi, w_{ij}) \not \equiv_P \false$
\end{itemize}
\end{lemma}

The following proposition constructs DBAs or DCAs for the fragments. The proof is an immediate consequence
of the lemma.

\begin{proposition}\label{prop:simpleaut}
Let $\varphi \in \ltlfum$. 
\begin{itemize}
\item The following DBA over the alphabet $2^{Ap}$ recognizes $L(\varphi)$: 
\[\mathcal{A}_\mu^\varphi= (\Reach(\varphi), \aft, \varphi, \inf(\true))\]
\item The following DBA over the alphabet $2^{Ap}$ recognizes $L(\G\F\varphi)$:
\[\mathcal{A}_{\G\F\mu}^\varphi= (\Reach(\F\varphi), \aft_{\F \varphi}, \F \varphi, \inf(\true))\]
\[\aft_{\F \varphi}(\psi, \nu) = \begin{cases} \F \varphi & \text{if } \psi \equiv_P \true \\ \aft(\psi, \nu) & \text{otherwise.} \end{cases}\]
\end{itemize}
Let $\varphi \in \ltlgrw$. 
\begin{itemize}
\item The following DCA over the alphabet $2^{Ap}$ recognizes $L(\varphi)$: 
\[\mathcal{A}_\nu^\varphi= (\Reach(\varphi), \aft, \varphi, \fin(\false))\]
\item The following DCA over the alphabet $2^{Ap}$ recognizes $L(\F\G\varphi)$:
\[\mathcal{A}_{\F\G\nu}^\varphi= (\Reach(\G\varphi), \aft_{\G \varphi}, \G \varphi, \fin(\false))\]
\[\aft_{\G \varphi}(\psi, \nu) = \begin{cases} \G \varphi & \text{if } \psi \equiv_P \false \\ \aft(\psi, \nu) & \text{otherwise.} \end{cases}\]
\end{itemize}
\end{proposition}

\begin{example}\label{ex:GF}
Let $\varphi =  a \wedge \X (b \vee \F c) \in \ltlfum$. The DBA  $\mathcal{A}_{\G\F\mu}^\varphi$ recognizing $L(\G\F\varphi)$ is depicted below. We use the abbreviations $\alpha:=\{\nu \in 2^{Ap} \mid a \in \nu\}$, $\beta:=\{\nu \in 2^{Ap} \mid b \in \nu\}$, and $\gamma:=\{\nu \in 2^{Ap} \mid c \in \nu\}$.\qee

\begin{center}
	\begin{tikzpicture}[x=3cm,y=1.5cm,font=\footnotesize,initial text=,outer sep=2pt]
	\tikzstyle{acc}=[double]

	\node[state,initial] (1) at (0,0) {$\F\varphi$};
	\node[state]         (2) at (1,0) {$\F\varphi \vee b \vee \F c$};
	\node[state]         (3) at (1,-0.75) {$\F\varphi \vee \F c$};
	\node[state,acc] 	 (4) at (2,0) {$\true$};
	
	\path[->] 
	(1) edge[loop above] node[above]{$2^{Ap} \setminus \alpha$} (1)
	(1) edge node[above]{$\alpha$} (2)
	(2) edge[loop above] node[above]{$\{a\}$} (2)
	(2) edge[bend right] node[left]{$\emptyset$} (3)
	(2) edge node[above]{$\beta, \gamma$} (4)
	(3) edge[bend right= 20] node[below,pos=0.1]{$\gamma$} (4)
	(3) edge[bend right] node[right]{$\{a\}, \{a, b\}$} (2)
	(3) edge[loop left] node[left]{$\emptyset, \{b\}$} (3)
	(4) edge[bend left=20, out=-45, in=225] node[above]{$2^{Ap}$} (1);
	\end{tikzpicture}
\end{center}
\end{example}

\begin{example}\label{ex:alternating}
Let $\varphi = a \W b \vee c \in \ltlgrw$. The DCA $\mathcal{A}_{\F\G\nu}^\varphi$ recognizing $L(\F\G\varphi)$ is depicted below. We use the abbreviations of Example \ref{ex:GF} again.

\begin{center}
	\begin{tikzpicture}[x=3cm,y=1.5cm,font=\footnotesize,initial text=,outer sep=2pt]
	\tikzstyle{acc}=[double]

	\node[state,initial] (1) at (0,0) {$\G \varphi$};
	\node[state] (2) at (1,0) {$\G \varphi \wedge a\W b$};
	\node[state,acc] (3) at (2,0) {$\false$};
	
	\path[->]  
	(1) edge[loop above] node[above]{$\beta, \gamma$} (1)
	(1) edge[bend left=10] node[above]{$\{a\}$} (2)
	(1) edge[bend right] node[below]{$\emptyset$} (3)
	(2) edge[loop above] node[above]{$\{a\},\{a,c\}$} (2)
	(2) edge[bend left=10] node[below]{$\beta$} (1)
	(2) edge node[above]{$\emptyset, \{c\}$} (3)
	(3) edge[bend left=20, out=-45, in=225] node[right, pos = 0.1]{$2^{Ap}$} (1);
	\end{tikzpicture}
\end{center}	
	Now consider the formula $\varphi = \F \G (a \U b \vee c)$. It does not belong to any of the fragments due to the deeper alternation of the least- and greatest-fixed-point operators:  $\F-\G-\U$. If we construct $\mathcal{A}_{\F\G\nu}^\varphi$ we obtain a DCA isomorphic to the one above, because $\aft(\psi_1 \U \psi_2, \nu)$ and  $\aft(\psi_1 \W \psi_2, \nu)$ are defined in the same way. However, the DCA does not recognize $L(\varphi)$: For example, on the word $\{a\}^\omega$, it loops on the middle state and accepts, even though $\{a\}^\omega \not\models \varphi$. 
The reason is that $\mathcal{A}_{\F\G\nu}^\varphi$ checks that the greatest fixed point holds, and cannot enforce satisfaction of the least-fixed-point formula $a\U b$. 
	
If only we were given a promise that $a\U b$ holds infinitely often, then we could conclude that such a run is accepting.
We can actually get such promises: for NBA and LDBA via the non-determinism of the automaton, and for DRA via the ``non-determinism'' of the acceptance condition. In the next section, we investigate how to utilize such promises (Section \ref{subsec:advice-use}) and how to check whether the promises are fulfilled or not (Section \ref{subsec:advice-check}). \qee
\end{example}

%% file: Sections/decomposition.tex
\section{The Master Theorem}
\label{sec:master}

We present and prove the Master Theorem: A characterization of the words satisfying a given formula from which we can easily extract deterministic, limit-deterministic, and nondeterministic automata of asymptotically optimal size. 

We first provide some intuition with the help of an example. Consider the formula $\varphi=\F\G((a\R b)\vee(c\U d))$, which does not belong to any of the fragments in the last section, and a word $w$. Assume we are promised that along $w$ the $\mu$-subformula $c \U d$ holds infinitely often (this is the case e.g. for $w=(\emptyset\{d\})^\omega$). In particular, we then know that $d$ holds infinitely often,  and so we can ``reduce''  $w\models^?\!\varphi$ to $w \models^?\! \F\G((a\R b)\vee(c\W d))$, which belongs to the fragment $\F\G(\ltlgrw)$.   

Assume now we are promised that $c \U d$ only holds finitely often (for example,  because $w=\{d\}^4\{c\}^\omega$). Even more, we are promised that along the suffix $w_5$ the formula $c \U d$ never holds any more. How can we use this advice?  First,  $w \models^? \!\varphi$ reduces to $w_5 \models^?\!  \aft(\varphi, w_{05})$ by the fundamental property of $\aft$, Lemma \ref{lem:af:prop}(1). Further, a little computation shows that $\aft(\varphi, w_{05}) \equiv_P  \varphi$, and so that $w \models^? \!\varphi$ reduces to $w_5 \models^?\! \varphi$. Finally, using that $c \U d$ never holds again, we reduce $w \models^? \!\varphi$ to $w_5 \models^? \F\G(a\R b \vee\false) \equiv_P \F\G(a\R b)$ which belongs to the fragment $\F\G(\ltlgrw)$.

This example suggests a general strategy for solving $w \models^?\! \varphi$: 
\begin{itemize}
\item Guess the set of least-fixed-point subformulas of $\varphi$ that hold infinitely often, denoted by $\setGF_w$, and the set of greatest-fixed-point subformulas that hold almost always, denoted by $\setFG_w$.
\item Guess a \emph{stabilization point} after which the least-fixed-point subformulas outside $\setGF_w$ do not hold any more, and the greatest-fixed-point subformulas of $\setFG_w$ hold forever.
\item Use these guesses to reduce $w \models^?\! \varphi$ to problems $w \models^?\! \psi$ for formulas $\psi$ that belong to the fragments introduced in the last section.
\item Check that the guesses are correct.
\end{itemize}
	
In the rest of the section we develop this strategy. In Section \ref{subsec:stability} we introduce the terminology
needed to formalize stabilization. Section \ref{subsec:advice} shows how to use a guess $X$ for $\setGF$ or a guess $Y$ for $\setFG$ to reduce $w \models^?\! \varphi$ to a simpler problem  $w \models^?\! \evalnu{\varphi}{X}$ or $w \models^?\! \evalmu{\varphi}{Y}$, where $\evalnu{\varphi}{X}$ and $\evalmu{\varphi}{Y}$ are read as ``$\varphi$ with $\G\F$-advice $X$'' and ``$\varphi$ with $\F\G$-advice $Y$'', respectively. Section \ref{subsec:advice-use} shows how to use the advice to decide $w \models^?\! \varphi$. Section \ref{subsec:advice-check} shows how to check that the advice is correct. The Master Theorem is stated and proved in Section \ref{subsec:advice-master-thm}.

\subsection{$\mu$- and $\nu$-stability.} \label{subsec:stability} 
Fix a formula $\varphi$. The set of subformulas of $\varphi $ of the form $\F\psi$, $\psi_1 \U \psi_2$, and $\psi_1 \M \psi_2$ is denoted by $\sfmu(\varphi)$. So, loosely speaking,  $\sfmu(\varphi)$ contains the set of subformulas of $\varphi$ with a least-fixed-point operator at the top of their syntax tree. Given a word $w$, we are interested in which of these formulas hold infinitely often, and which ones hold at least once, i.e., we are interested in the sets 
\begin{align*} 
\setGF_w & = \{\psi \mid \psi \in \sfmu(\varphi) \wedge w \models \G\F\psi \} \\
\setF_w & = \{\psi \mid \psi \in \sfmu(\varphi) \wedge w \models \F\psi \} 
\end{align*}
\noindent Observe that $\setGF_w \subseteq \setF_w$. We say that $w$ is $\mu$-stable with respect to 
$\varphi$ if $\setGF_w = \setF_w$.

\begin{example}
\label{ex:sfmu}
For $\varphi = \G a \vee b \U c$ we have $\sfmu(\varphi) = \{ b \U c \}$. Let $w = \{a\}^\omega$ and $w'=\{b\}\{c\}\{a\}^\omega$. We have $\setF_{w}= \emptyset = \setGF_{w}$ and $\setGF_{w'}= \emptyset \subset \{b \U c\} = \setF_{w'}$. So $w$ is $\mu$-stable with respect to $\varphi$, but $w'$ is not. \qee
\end{example}

Dually, the set of subformulas of $\varphi $ of the form $\G\psi$, $\psi_1 \W \psi_2$, and $\psi_1 \R \psi_2$ is denoted by $\sfnu(\varphi)$. This time we are interested in whether these formulas hold everywhere or almost everywhere, i.e., in the sets 
\begin{align*} 
\setFG_w & = \{\psi \mid \psi \in \sfnu(\varphi) \wedge w \models \F\G\psi \} \\
\setG_w & = \{\psi \mid \psi \in \sfnu(\varphi) \wedge w \models \G\psi \} 
\end{align*}
\noindent (Observe that the question whether a $\nu$-formula like, say, $\G a$, holds once or infinitely often makes no sense, because it holds once if{}f it holds infnitely often.) 
We have $\setFG_w \supseteq \setG_w$, and we say that  $w$ is $\nu$-stable with respect to 
$\varphi$ if $\setFG_w = \setG_w$.

\begin{example}
Let $\varphi$, $w$ and $w'$ as in Example \ref{ex:sfmu}. We have $\sfnu(\varphi) = \{ \G a \}$. The word $w$ is $\nu$-stable, but $w'$ is not, because $\setFG_{w'}= \{ \G a \} \supset \emptyset = \setG_{w'}$. \qee
\end{example}

So not every word is $\mu$-stable or $\nu$-stable. However, as shown by the following lemma,  all but finitely many suffixes of a word are $\mu$- and $\nu$-stable. 

\begin{lemma}
\label{lem:stable}
For every word $w$ there exist indices $i,j \geq 0$ such that for every $k \geq 0$ the 
suffix $w_{i+k}$ is $\mu$-stable and the suffix $w_{j+k}$ is $\nu$-stable.
\end{lemma}
\begin{proof}
We only prove the $\mu$-stability part; the proof of the other part is similar. Since $\setGF_{w_i} \subseteq \setF_{w_i}$ for every $i \geq 0$,
it suffices to exhibit an index $i$ such that $\setGF_{w_{i+k}} \supseteq \setF_{w_{i+k}}$ for every $k \geq 0$. If  $\setGF_{w} \supseteq \setF_{w}$ then we can choose $i:=0$. 
So assume $\setF_{w} \setminus \setGF_{w} \neq \emptyset$. By definition, every 
$\psi \in \setF_w \setminus \setGF_{w}$ holds only finitely often along $w$. So for every
$\psi \in \setF_w \setminus \setGF_{w}$ there exists an index $i_\psi$ such that $w_{i_\psi +k} \not\models  \psi$ for every $k \geq 0$. Let $i := \max \{i_\psi \mid \psi \in \setF_w\}$, which exists because  $\setF_w$ is a finite set. It follows $\setGF_{w_{i+k}} \supseteq \setF_{w_{i+k}}$ for every $k \geq 0$, and so every $w_{i+k}$ is $\mu$-stable.
\end{proof}

\begin{example}
Let again $\varphi = \G a \vee b \U c$. The word $w' = \{b\}\{c\}\{a\}^\omega$ is neither $\mu$-stable nor $\nu$-stable, but all suffixes $w'_{(2+k)}$ of $w'$ are both $\mu$-stable and $\nu$-stable. \qee
\end{example}

\subsection{The formulas $\evalnu{\varphi}{X}$ and $\evalmu{\varphi}{Y}$.} \label{subsec:advice} 

We first introduce $\evalnu{\varphi}{X}$. 
Assume we have to determine if a word $w$ satisfies $\varphi$, and we are told that $w$ is $\mu$-stable. Further, we are given the set $X \subseteq \sfmu(\varphi)$ such that $\setGF_w = X = \setF_w$. We use this oracle information to reduce the problem $w \models^? \varphi$ to a ``simpler'' problem $w \models^? \evalnu{\varphi}{X}$, where ``simpler'' means that $\evalnu{\varphi}{X}$ is a formula of $\ltlgrw$, for which we already know how to construct automata. In other words, we define a formula $\evalnu{\varphi}{X} \in \ltlgrw$ such that  $\setGF_w = X = \setF_w$ implies $w \models \varphi$ if{}f $w \models \evalnu{\varphi}{X}$. (Observe that $X \subseteq \sfmu(\varphi)$ but $\evalnu{\varphi}{X} \in \ltlgrw$, and so the latter, not the former, is the reason for the $\nu$-subscript in the notation $\evalnu{\varphi}{X}$.)

The definition of $\evalnu{\varphi}{X}$ is purely syntactic, and
the intuition behind it is very simple. All the main ideas are illustrated by the following examples, where we assume $\setGF_w = X = \setF_w$:
\begin{itemize}
\item $\varphi = \F a \wedge \G b$ and $X = \{ \F a \}$. Then $\F a \in \setGF_w$, which implies in particular $w \models \F a$. So we can reduce  $w \models^? \F a \wedge \G b$ to $w \models^?  \G b$, and so $\evalnu{\varphi}{X}:= \G b$.
\item $\varphi = \F a \wedge \G b$ and $X = \emptyset$. Then $\F a \not\in \setF_w$, and so $w \not\models \F a$. So we can reduce  $w \models^? \F a \wedge \G b$ to the trivial problem $w \models^?  \false$, and so $\evalnu{\varphi}{X}:= \false$.
\item $\varphi = \G(b \U c)$ and $X=\{ b \U c\}$. Then $b \U c \in \setGF_w$, and so $w \models \G\F (b \U c)$. This does not imply $w \models b \U c$, but implies that $c$ will hold in the future. So we can reduce $w \models^?  \G(b \U c)$ to $w \models^? \G(b \W c)$, a formula of $\ltlgrw$, and so $\evalnu{\varphi}{X}:= \G(b \W c)$.
\end{itemize}

\begin{definition}\label{def:evalnu}
Let $\varphi$ be a formula and let $X \subseteq \sfmu(\varphi)$. The formula $\evalnu{\varphi}{X}$ is inductively defined as follows:
\begin{itemize}
\item If $\varphi = \true, \false, a, \neg a$, then $\evalnu{\varphi}{X}=\varphi$.
\item If $\varphi = \mathit{op}(\psi)$ for $\mathit{op} \in \{ \X, \G \}$ then $\evalnu{\varphi}{X}=\mathit{op}(\evalnu{\psi}{X})$.
\item If $\varphi = \mathit{op}(\psi_1, \psi_2)$ for $\mathit{op} \in \{ \wedge, \vee, \W, \R \}$ then \\ $\evalnu{\varphi}{X}= \mathit{op}(\evalnu{\psi_1}{X}, \evalnu{\psi_2}{X})$.
\item If $\varphi = \F\psi$ then		
$\evalnu{\varphi}{X} = \begin{cases} \true & \mbox{if $\varphi \in X$} \\ \false & \mbox{otherwise.}\end{cases} $
\item If $\varphi = \psi_1 \U \psi_2$ then 
$\evalnu{\varphi}{X} =  \begin{cases} (\evalnu{\psi_1}{X})\W(\evalnu{\psi_2}{X}) & \mbox{if $\varphi \in X$} \\ \false & \mbox{otherwise.}\end{cases}$
\item If $\varphi = \psi_1 \M \psi_2$ then 
$\evalnu{\varphi}{X} =  \begin{cases} (\evalnu{\psi_1}{X})\R(\evalnu{\psi_2}{X}) & \mbox{if $\varphi \in X$} \\ \false & \mbox{otherwise.}\end{cases}$
\end{itemize}
\end{definition}

We now introduce, in a dual way, a formula  $\evalmu{\varphi}{Y} \in \ltlfum$ such that  $\setFG_w=Y=\setG_w$ implies  $w \models \varphi$ if{}f $w \models \evalmu{\varphi}{Y}$.

\begin{definition}\label{def:evalmu}
Let $\varphi$ be a formula and let $Y \subseteq \sfnu(\varphi)$. The formula $\evalmu{\varphi}{Y}$ is inductively defined as follows:
\begin{itemize}
\item If $\varphi = \true, \false, a, \neg a$, then $\evalmu{\varphi}{Y}=\varphi$.
\item If $\varphi = \mathit{op}(\psi)$ for $\mathit{op} \in \{ \X, \F \}$ then $\evalmu{\varphi}{Y}=\mathit{op}(\evalmu{\psi}{Y})$.
\item If $\varphi = \mathit{op}(\psi_1, \psi_2)$ for $\mathit{op} \in \{ \wedge, \vee, \U, \M \}$ then \\ $\evalmu{\varphi}{Y}= \mathit{op}(\evalmu{\psi_1}{Y}, \evalmu{\psi_2}{Y})$.
\item If $\varphi = \G\psi$ then		
$\evalmu{\varphi}{Y} = \begin{cases} \true & \mbox{if $\varphi \in Y$} \\ \false & \mbox{otherwise.}\end{cases} $
\item If $\varphi = \psi_1 \W \psi_2$ then 
$\evalmu{\varphi}{Y} = \begin{cases} \true & \mbox{if $\varphi \in Y$} \\ (\evalmu{\psi_1}{Y})\U(\evalmu{\psi_2}{Y}) & \mbox{otherwise.}\end{cases}$
\item If $\varphi = \psi_1 \R \psi_2$ then 
$\evalmu{\varphi}{Y} = \begin{cases} \true & \mbox{if $\varphi \in Y$} \\ (\evalmu{\psi_1}{Y})\M(\evalmu{\psi_2}{Y}) & \mbox{otherwise.}\end{cases}$
\end{itemize}
\end{definition}

\begin{example}
Let $\varphi = ((a \W b) \wedge \F c) \vee a \U d$.  We have:
\[\begin{array}{lclcl}
\evalnu{\varphi}{\{\F c\}} & = & ((a \W b) \wedge \true) \vee \false & \equiv_P & a \W b  \\
\evalnu{\varphi}{\{a \U d\}} & = & ((a \W b) \wedge \false) \vee a \W d & \equiv_P & a \W d \\
\evalnu{\varphi}{\emptyset} & = & ((a \W b) \wedge \false) \vee \false & \equiv_P & \false  \\
\evalmu{\varphi}{\{a \W b\}} & = & (\true \wedge \F c) \vee a \U d & \equiv_P &  \F c \vee a \U d \\
\evalmu{\varphi}{\emptyset} & = &  (a \U b \wedge \F c) \vee a \U d \\
\end{array}\] \qee
\end{example}

\subsection{Utilizing $\evalnu{\varphi}{X}$ and $\evalmu{\varphi}{Y}$.} \label{subsec:advice-use}

The following lemma states the fundamental properties of $\evalnu{\varphi}{X}$ and $\evalmu{\varphi}{Y}$. As announced above, for a $\mu$-stable word $w$ we can reduce the problem $w \models^? \varphi$ to  $w \models^? \evalnu{\varphi}{X}$, and for a $\nu$-stable word to $w \models^? \evalmu{\varphi}{Y}$. However, there is more: If we only know $X \subseteq \setGF_w$, then we can still infer $w \models \varphi$ from $w \models \evalnu{\varphi}{X}$, only the implication in the other direction fails.  

\begin{restatable}{lemma}{lemEvalnuEvalmu}\label{lem:evalnu,evalmu}
Let $\varphi$ be a formula and let $w$ be a word.\\
For every $X \subseteq \sfmu(\varphi)$:
\begin{itemize}
\item[(a1)] If $~\setF_w \subseteq X$ and $w \models \varphi$, then $w \models \evalnu{\varphi}{X}$.
\item[(a2)] If $~X \subseteq \setGF_w$ and $w \models \evalnu{\varphi}{X}$, then $w \models \varphi$.
\end{itemize}
\noindent In particular:
\begin{itemize}
\item[(a3)] If $~\setF_w = X = \setGF_w$ then $w \models \varphi$ if{}f $w \models \evalnu{\varphi}{X}$.
\end{itemize}
For every $Y \subseteq \sfnu(\varphi)$:
\begin{itemize}
\item[(b1)] If $~\setFG_w \subseteq Y$ and $w \models \varphi$, then $w \models \evalmu{\varphi}{Y}$.
\item[(b2)] If $~Y \subseteq \setG_w$ and  $w \models \evalmu{\varphi}{Y}$, then $w \models \varphi$.
\end{itemize}
\noindent In particular:
\begin{itemize}
\item[(b3)] If $~\setFG_w = Y = \setG_w$ then  $w \models \varphi$ if{}f $w \models \evalmu{\varphi}{Y}$.
\end{itemize}
\end{restatable}

\begin{proof} 
All parts are proved by a straightforward structural induction on $\varphi$. We consider only (a1), and only two representative cases of the induction. Representative cases for (a2), (b1), and (b2) can be found in the appendix of \cite{arXiv}. 

\smallskip

\noindent (a1) Assume $~\setF_w \subseteq X$. Then $\setF_{w_i} \subseteq X$ for all $i \geq 0$. We prove the following stronger statement via structural induction on $\varphi$:
\[\forall i. ~ (\; (w_i \models \varphi) \rightarrow (w_i \models \evalnu{\varphi}{X}) \; )\]

We consider one representative of the \enquote{interesting} cases, and one of the \enquote{straightforward} cases.

\noindent Case $\varphi = \psi_1 \U \psi_2$: Let $i \geq 0$ arbitrary and assume $w_i \models \psi_1 \U \psi_2$. Then $\psi_1 \U \psi_2 \in \setF_{w_i}$ and so $\varphi \in X$. We prove $w_i \models \evalnu{(\psi_1 \U \psi_2)}{X}$: 
\[
\arraycolsep=1.8pt
\begin{array}{clr}
 	        & w_i \models \psi_1 \U \psi_2 \\
\implies  & w_i \models \psi_1 \W \psi_2 \\ 
\implies  & \forall j. ~ w_{i+j} \models \psi_1 \vee \exists k \leq j.~ w_{i+k} \models \psi_2 \\
\implies  & \forall j. ~ w_{i+j} \models \evalnu{\psi_1}{X} \vee \exists k \leq j.~ w_{i+k} \models \evalnu{\psi_2}{X} & \text{(I.H.)} \\
\implies  & w_i \models (\evalnu{\psi_1}{X}) \W (\evalnu{\psi_2}{X}) \\
\implies  & w_i \models \evalnu{(\psi_1 \U \psi_2)}{X} & \text{($\varphi \in X$, Def. \ref{def:evalnu})}
\end{array}
\]

\noindent Case $\varphi = \psi_1 \vee \psi_2$: Let $i \geq 0$ arbitrary and assume $w_i \models \psi_1 \vee \psi_2$: 
\[\begin{array}{clrl}
	        & w_i \models \psi_1 \vee \psi_2  \\
\implies \; & (w_i \models \psi_1) \vee (w_i \models \psi_2) \\
\implies \; & (w_i \models \evalnu{\psi_1}{X}) \vee (w_i \models \evalnu{\psi_2}{X}) &  \text{(I.H.)} \\
\implies \; & w_i \models \evalnu{(\psi_1 \vee \psi_2)}{X} & \text{(Def. \ref{def:evalnu})} & \qedhere
\end{array}\]
\end{proof}

Lemma \ref{lem:evalnu,evalmu} suggests to decide $w \models^? \varphi$ by ``trying out'' all possible sets $X$. Part (a2) shows that the strategy of checking for every set $X$ if  both $X \subseteq \setGF_w$ and $w \models \evalnu{\varphi}{X}$ hold is sound. 

\begin{example}
\label{ex:possX}
Consider $\varphi=\G\F a\vee \G\F(b\wedge \G c)$. Since $\sfmu(\varphi) = \{ \F a, \F(b\wedge \G c)\}$, there are four possible $X$'s to be tried out: $\emptyset$, $\{\F a\}$, $\{\F(b \wedge \G c)\}$, and $\{\F a, \F(b \wedge \G c)\}$.  For $X = \emptyset$ we get $\evalnu{\varphi}{X} = \false$, indicating that if neither $a$ nor $b \wedge \G c$ hold infinitely often, then $\varphi$ cannot hold. For the other three possibilities ($a$ holds infinitely often, $b\wedge \G c$ holds infinitely often, or both) there are words satisfying $\varphi$, like $a^\omega$, $\{b,c\}^\omega$, and $\{a,b,c\}^\omega$. \qee
\end{example}

However there are still two questions open. First, is this strategy complete? Part (a3) shows that it is complete for $\mu$-stable words: Indeed, in this case there is a set $X$ such that $\setGF_w = X = \setF_w$, and for this particular set
$w \models \evalnu{\varphi}{X}$ holds. For words that are not $\mu$-stable, we will use the existence of $\mu$-stable suffixes: Instead of checking $w \models \evalnu{\varphi}{X}$, we will check the existence of a suffix $w_i$ such that $w_{i} \models \evalnu{\aft(\varphi, w_{0i})}{X}$. This will happen in Section \ref{subsec:advice-master-thm}. 
The second open question is simply how to check $X \subseteq \setGF_w$. We deal with it in
Section  \ref{subsec:advice-check}.

\subsection{Checking $X \subseteq \setGF_w$ and $Y \subseteq \setFG_w$.} \label{subsec:advice-check}

Consider again the formula $\varphi=\G\F a\vee \G\F(b\wedge \G c)$ of Example \ref{ex:possX}.
If $X=\{\F a\}$, then checking whether $X$ is a correct advice (i.e., whether $X \subseteq \setGF_w$ holds) is easy, because $\G\F\F a \in \G\F(\ltlfum)$, see Proposition~\ref{prop:simpleaut}. In contrast, for $X=\{\F(b\wedge \G c)\}$ this is not so.
In this case it would come handy if we had an advice $Y= \{\G c\}$ promising that $\G c$ holds almost always, as is the case for e.g. $\emptyset^5(\{b,c\}\{c\})^\omega$. Indeed, we could easily check correctness of this advice, because $\F\G\G c \in \F\G(\ltlgrw)$, and with its help checking $\G\F(b\wedge \G c)$ reduces to checking $\G\F(b\wedge \true)=\G\F b$, which is also easy.

One of the main ingredients of our approach is that in order to verify a promise $X \subseteq \setGF_w$ we can rely on a promise  $Y \subseteq \setFG_w$ about subformulas of $X$, and vice versa. There is no circularity in this rely/guarantee reasoning because the subformula order is well founded, and we eventually reach formulas $\psi$ such that $\evalnu{\psi}{X}=\psi$ or $\evalmu{\psi}{Y}=\psi$. This argument is formalized in the next lemma. The first part of the lemma states that mutually assuming correctness of  the other promise is correct. The second part states that, loosely speaking, this rely/guarantee method is complete.%: it can prove that $X=\setGF_w$ and $Y = \setFG_w$ hold.

\begin{lemma}\label{lem:evalFG:setFG}
Let $\varphi$ be a formula and let $w$ be a word. 
\begin{itemize}
\item[(1.)]  For every  $X \subseteq \sfmu(\varphi)$ and $Y \subseteq \sfnu(\varphi)$, if  
\[\begin{array}{l}
\forall \psi \in X. \; w \models \G\F(\evalmu{\psi}{Y})\\
\forall \psi \in Y. \; w \models \F\G(\evalnu{\psi}{X})
\end{array}\]
\noindent  then  $X \subseteq \setGF_w$ and $Y \subseteq \setFG_w$.
\smallskip
\item[(2.)] If  $X = \setGF_w$ and $Y = \setFG_w$ then
\[\begin{array}{l}
\forall \psi \in X. \; w \models \G\F(\evalmu{\psi}{Y})\\
\forall \psi \in Y. \; w \models \F\G(\evalnu{\psi}{X})
\end{array}\]
\end{itemize}
\end{lemma}
\begin{proof} \noindent
(1.) Let $X \subseteq \sfmu(\varphi)$ and $Y \subseteq \sfnu(\varphi)$. Observe that $X \cap Y = \emptyset$. Let $n:= |X \cup Y|$. Let $\psi_1, \ldots, \psi_n$ be an enumeration of $X \cup Y$ compatible with the subformula order, i.e., if $\psi_i$ is a subformula of $\psi_j$, then $i \leq j$. Finally, let $(X_0, Y_0), (X_1, Y_1), \ldots, (X_{n}, Y_{n})$ be the unique sequence of
pairs satisfying:
\begin{itemize}
\item $(X_0, Y_0) = (\emptyset, \emptyset)$ and  $(X_n, Y_n) = (X, Y)$. 
\item For every $0 < i \leq n$, if $\psi_i \in X$ then $X_{i} \setminus X_{i-1} = \{ \psi_i \}$ and $Y_i = Y_{i-1}$, and if $\psi_i \in Y$, then $X_i = X_{i-1}$ and $Y_{i} \setminus Y_{i-1} = \{ \psi_i \}$.
\end{itemize}

We prove $X_i \subseteq \setGF_w$ and $Y_i \subseteq \setFG_w$ for every $0 \leq i \leq n$ by induction on $i$. For $i=0$ the result follows immediately from $X_0 = \emptyset= Y_0$. For $i > 0$ we consider two cases: 

\smallskip

\noindent \textbf{Case 1:} $\psi_{i} \in Y$, i.e., $X_i = X_{i-1}$ and $Y_i \setminus Y_{i-1} = \{ \psi_i \}$. 

By induction hypothesis and $X_{i} = X_{i-1}$ we have $X_{i} \subseteq \setGF_w$ and $Y_{i-1} \subseteq \setFG_w$. We prove $\psi_i \in \setFG_w$, i.e., $w \models \F\G \psi_i$, in three steps.

\noindent \textbf{Claim 1:} $\evalnu{\psi_i}{X} = \evalnu{\psi_i}{X_{i}}$. \\
By the definition of  the $\evalnu{\cdot}{\cdot}$ mapping, $\evalnu{\psi_i}{X}$ is completely determined by the $\mu$-subformulas of $\psi_i$ that belong to $X$. By the definition of the sequence
$(X_0, Y_0), \ldots, (X_{n}, Y_{n})$, a $\mu$-subformula of $\psi_i$ belongs to $X$ if{}f it belongs to $X_{i}$, and we are done. 

\noindent \textbf{Claim 2:} $X_i \subseteq \setGF_{w_{k}}$ for every $k \geq 0$. \\
Follows immediately from $X_{i} \subseteq \setGF_w$.

\noindent \textbf{Proof of $w \models \F\G \psi_i$.} 
By the assumption of the lemma we have $w \models \F\G(\evalnu{\psi_i}{X})$, and so, by Claim~1,  $w \models \F\G(\evalnu{\psi_i}{X_{i}})$. So there exists an index $j$
such that $w_{j+k} \models \evalnu{\psi_i}{X_i}$ for every $k \geq 0$. By Claim~2 we further have $X_i \subseteq \setGF_{w_{j+k}}$ for every $j, k \geq 0$. So we can apply part (a2) of Lemma \ref{lem:evalnu,evalmu} to $X_i$, $w_{j+k}$, and $\psi_i$, which yields
$w_{j+k} \models \psi_i$ for every $k \geq 0$. So $w \models \F\G \psi_i$.

\smallskip

\noindent \textbf{Case 2:} $\psi_{i} \in X$, i.e., $X_i \setminus X_{i-1} = \{ \psi_i \}$ and $Y_i = Y_{i-1}$. \\
In this case $X_{i-1} \subseteq \setGF_w$ and $Y_{i} \subseteq \setFG_w$. We prove $\psi_i \in \setGF_w$, i.e., $w \models \G\F \psi_i$ in three steps. 

\noindent \textbf{Claim 1:} $\evalmu{\psi_i}{Y} = \evalmu{\psi_i}{Y_{i}}$. \\
The claim is proved as in Case 1.

\noindent \textbf{Claim 2:} There is an $j \geq 0$ such that $Y_i \subseteq \setG_{w_{k}}$ for every $k \geq j$. \\
Follows immediately from $Y_{i} \subseteq \setFG_w$.

\noindent \textbf{Proof of $w \models \G\F \psi_i$.} By the assumption of the lemma we have $w \models \G\F(\evalmu{\psi_i}{Y})$. Let $j$ be the index of Claim~2.
By Claim~1 we have $w \models \G\F(\evalmu{\psi_i}{Y_{i}})$, and so there exist infinitely many 
$k \geq j$ such that $w_k \models \evalmu{\psi_i}{Y_{i}}$. By Claim~2 we further have 
$Y_i \subseteq \setG_{w_{k}}$. So we 
can apply part (b2) of Lemma \ref{lem:evalnu,evalmu} to $Y_i$, $w_{k}$, and $\psi_i$, which yields 
$w_{k} \models \psi_i$ for infinitely many $k \geq j$. So $w \models \G\F \psi_i$.

\medskip

\noindent (2.) Let $\psi \in \setGF_w$. We have $w \models \G\F\psi$, and so $w_i \models \psi$ for infinitely many $i \geq 0$. Since $\setFG_{w_i} = \setFG_w$ for every $i \geq 0$, part (b1) of Lemma \ref{lem:evalnu,evalmu} can be applied to $w_i$, $\setFG_{w_i}$, and $\psi$. This yields $w_i \models \evalmu{\psi}{\setFG_w}$ for infinitely many $i \geq 0$ and thus $w \models \G\F(\evalmu{\psi}{\setFG_w})$.

\smallskip

Let $\psi \in \setFG_w$. Since $w_i \models \F\G\psi$, there is an index $j$ such that  $w_{j+k} \models \psi$ for every $k \geq 0$. By Lemma \ref{lem:stable} the index $j$ can be chosen so that it also satisfies $\setGF_{w} = \setF_{w_{j+k}} = \setGF_{w_{j+k}}$ for every $k\geq 0$.  So part (a1) of Lemma \ref{lem:evalnu,evalmu} can be applied to $\setF	_{w_{j+k}}$, $w_{j+k}$, and $\psi$. This yields $w_{j+k} \models \evalnu{\psi}{\setGF_w}$ for every $k \geq 0$ and thus $w \models \F\G(\evalnu{\psi}{\setGF_w})$.
\end{proof}

\begin{example}
\label{ex:mutual}
Let $\varphi = \F(a \wedge \G( b \vee \F c))$, $X = \{ \varphi \}$, and 
$Y = \{ \G( b \vee \F c) \}$. 
\begin{itemize}
\item The condition $\forall \psi \in X. \; w \models \G\F(\evalmu{\psi}{Y})$ becomes
\[w \models \G\F\left(\evalmu{\varphi}{Y}\right) = \G\F(\F a) \equiv \G\F a\]
\item The condition $\forall \psi \in Y. \; w \models \F\G(\evalnu{\psi}{X})$ becomes
\[w  \models  \F\G\left(\evalnu{ \G( b \vee \F c)}{X}\right) = \F\G (\G b ) \equiv \F\G b\]
\end{itemize}
By Lemma \ref{lem:evalFG:setFG} (1) we then have that $w \models \G\F a \wedge \F\G b$ implies
$\varphi \in \setGF_w$ and $\G(b \vee \F c)) \in \setFG_w$. \qee
\end{example}

\subsection{Putting the pieces together: The Master Theorem.} \label{subsec:advice-master-thm}

Putting together Lemma \ref{lem:evalnu,evalmu} and Lemma \ref{lem:evalFG:setFG},
we obtain the main result of the paper, which we will use as ``Master Theorem'' for the construction of automata.

\begin{theorem}[Master Theorem]\label{thm:logical:char} For every formula $\varphi$ and for every word $w$: \, $w \models \varphi$ if{}f there exists $X \subseteq \sfmu(\varphi)$ and $Y \subseteq \sfnu(\varphi)$ satisfying
\begin{align*}
(1) & \quad \exists i. ~ w_{i} \models \evalnu{\aft(\varphi, w_{0i})}{X} \\
(2) & \quad \forall \psi \in X. ~ w \models \G\F(\evalmu{\psi}{Y})  \\
(3) & \quad \forall \psi \in Y. ~ w \models \F\G(\evalnu{\psi}{X})
\end{align*}
\end{theorem}

Observe that $\evalnu{\aft(\varphi, w_{0i})}{X}$,  $\G\F(\evalmu{\psi}{Y})$, and 
$\F\G(\evalnu{\psi}{X})$ are formulas of $\ltlgrw$, $\G\F(\ltlfum)$, and $\F\G(\ltlgrw)$,
respectively, i.e., they all belong to the fragments of Section \ref{sec:fragments}.

Before proving the theorem, let us interpret it in informal terms. The Master Theorem states that
in order to decide $w \models^?\! \varphi$ we can guess two sets $X \subseteq \sfmu(\varphi)$
and $Y \subseteq \sfnu(\varphi)$ and an index $i$, and then proceed as follows: verify $Y \subseteq \setFG_w$  assuming that $X \subseteq \setGF_w$ holds (3), verify $X \subseteq \setGF_w$  assuming that $Y \subseteq \setFG_w$ holds (2), and verify $w_i \models \aft(\varphi, w_{0i})$ assuming that $X \subseteq \setGF_w$ holds (1). The procedure is sound by Lemma \ref{lem:evalnu,evalmu} and Lemma \ref{lem:evalFG:setFG}, and complete because the guess where $X := \setGF_w$, $Y := \setFG_w$, and $i$ is a stabilization point of $w$, is guaranteed to succeed.

\begin{example}
Let $\varphi = \F(a \wedge \G( b \vee \F c))$ as in Example \ref{ex:mutual}, and let $\varphi' = d \U \varphi$. For $X = \{\varphi,\varphi' \}$,   $Y = \{ \G( b \vee \F c) \}$, and $i=0$ the Master Theorem yields that $w \models \varphi'$ is implied by 
\begin{itemize}
\item[(1)] $w \models \evalnu{(d \U \varphi)}{X}= d \W (\evalnu{\varphi}{X}) = d \W \true \equiv \true$, 
\item[(2)] $w \models \G\F(\evalmu{\varphi}{Y}) \wedge \G\F(\evalmu{\varphi'}{Y}) =  \G\F a \wedge \G\F (d \U (\F a))$, and 
\item[(3)] $w \models \F\G(\evalnu{(\G(b \vee \F c)}{X}) \equiv \F\G b$. 
\end{itemize}
For $X=\{\varphi\}$, $Y = \{ \G( b \vee \F c) \}$, and $i=0$, condition (1)  is $w \models \false$, and we do not derive any useful information. \qee
\end{example}

\smallskip

\begin{proof}[Proof (of the Master Theorem)] ~\\ 
\noindent ($\Rightarrow$): Assume $w \models \varphi$, and set $X := \setGF_w$ and $Y := \setFG_w$. Properties (2) and (3) follow from Lemma \ref{lem:evalFG:setFG}. For property (1), let $i$ be an index such that $\setF_{w_i} = \setGF_{w_i}$; this index exists by Lemma \ref{lem:stable}. By Lemma \ref{lem:af:prop} we have $w_i \models \aft(\varphi, w_{0i})$, and by Lemma \ref{lem:evalnu,evalmu} (a1) $w_i \models \evalnu{\aft(\varphi, w_{0i})}{X}$.

\smallskip

\noindent ($\Leftarrow$): Assume that properties (1-3) hold for sets $X,Y$ and an index $i$. By Lemma \ref{lem:evalFG:setFG} (1.) we have $X \subseteq \setGF_w$, and so
$X \subseteq \setGF_{w_i}$. By Lemma \ref{lem:evalnu,evalmu} (a2) we obtain $w_i \models \evalnu{\aft(\varphi,w_{0i})}{X}$, and thus $w_i \models \aft(\varphi, w_{0i})$. Lemma \ref{lem:af:prop} yields $w \models \varphi$.
\end{proof}

Let $L_{X,Y}^j$ be the language of all words that satisfy condition (j) of the Master Theorem for the sets $X$ and $Y$. The Master Theorem can then be reformulated as:
\[L(\varphi) = \bigcup_{\substack{X \subseteq \sfmu(\varphi)\\ Y \subseteq \sfnu(\varphi)}} L_{X,Y}^1 \cap L_{X,Y}^2 \cap L_{X,Y}^3\]

\noindent Therefore, given an automata model effectively closed under union and intersection,  in order to construct automata for all of LTL it suffices to exhibit automata recognizing $L_{X,Y}^1, L_{X,Y}^2, L_{X,Y}^3$.  In the next section we consider the case of DRAs, and then we proceed to NBAs and LDBAs. 

%% file: Sections/construction-dra.tex
\section{Constructing DRAs for LTL Formulas}
\label{sec:ltl2dra}

Let $\varphi$ be a formula of length $n$. We use the Master Theorem to construct a DRA
for $L(\varphi)$ with $2^{2^{O(n)}}$ states and $O(2^n)$ Rabin pairs. Since our purpose is only to show that we can easily obtain automata of asymptotically optimal size, we give priority to a simpler construction over one with the least number of states. We comment in Section \ref{sec:discussion} on optimizations that reduce the size by using other acceptance conditions.

We first construct DRAs for $L_{X,Y}^1$, $L_{X,Y}^2$, and  $L_{X,Y}^3$ with $2^{2^{O(n)}}$ states and one single Rabin pair. More precisely, for each of these languages we construct either a DBA or a DCA. We then construct a DRA for $L(\varphi)$ by means of intersections and unions. 

\paragraph{A DCA for $L_{X,Y}^1$.} We define a DCA $\mathcal{C}_{\varphi,X}$ that accepts a word $w$ if{}f $w_{i} \models \evalnu{\aft(\varphi, w_{0i})}{X}$ for some suffix $w_i$ of $w$.  In the rest of this part of the section we abbreviate $\aft(\varphi, w_{0i})$ to $\varphi_i$.  Recall that $\evalnu{\varphi_i}{X}$ is a formula of $\ltlgrw$, and so for every $i \geq 0$ there is a DCA with a state $\false$ such that the automaton rejects if{}f it reaches this state. Intuitively, if the automaton rejects, then it rejects ``after finite time''.
We prove the following lemma:

\begin{lemma}\label{lem:delay}
Let $\varphi_i := \aft(\varphi, w_{0i})$. If $w \models \evalnu{\varphi}{X}$ then $w_i \models \evalnu{\varphi_i}{X}$ for all $i > 0$.
\end{lemma}

\begin{proof}
Assume $w \models \evalnu{\varphi}{X}$. It suffices to prove $w_1 \models \evalnu{\varphi_1}{X}$, since the general case follows immediately by induction. For $i = 1$ we proceed by structural induction on $\varphi$, and consider only some representative cases.

\noindent \textbf{Case $\varphi = a$.}  Since $w \models \evalnu{a}{X} = a$ we have $a \in w[0]$. So $\evalnu{\varphi_1}{X} = \evalnu{\true}{X} = \true$, and thus $w_1 \models\evalnu{\varphi_1}{X}$.

\noindent \textbf{Case $\varphi = \psi \U \chi$.}  Since  $w \models \evalnu{\varphi}{X}$ we have $\evalnu{\varphi}{X} \neq \false$, and so $\varphi \in X$. We have:
\[
\arraycolsep=1.8pt
\begin{array}{clr}
	     & w \models \evalnu{\varphi}{X} \\
\implies & w \models (\evalnu{\psi}{X}) \W (\evalnu{\chi}{X}) & \text{(Def. \ref{def:evalnu})} \\
\implies & w \models (\evalnu{\psi}{X} \wedge \X ((\evalnu{\psi}{X}) \W (\evalnu{\chi}{X}))) \vee \evalnu{\chi}{X} \\
\implies & w \models (\evalnu{\psi}{X} \wedge \X (\evalnu{(\psi \U \chi)}{X})) \vee \evalnu{\chi}{X} & \text{($\varphi \in X$)} \\
\implies & w_1 \models (\evalnu{\psi_1}{X} \wedge \evalnu{\varphi}{X}) \vee \evalnu{\chi_1}{X} & \text{(I.H.)} \\
\implies & w_1 \models \evalnu{((\psi_1 \wedge (\psi \U \chi)) \vee \chi_1)}{X} & \text{(Def. \ref{def:evalnu})}  \\
\implies & w_1 \models \evalnu{\varphi_1}{X} & \text{(Def. \ref{def:af})} 
\end{array}\] \qedhere
\end{proof}

\noindent Loosely speaking,  $\mathcal{C}_{\varphi,X}$ starts by checking $w~\models^?~\evalnu{\varphi}{X}$. For this it maintains the formula $(\evalnu{\varphi}{X})_i$ in its state. If the formula becomes  $\false$ after, say, $j$ steps, then  $w~\not\models~\evalnu{\varphi}{X}$, and $\mathcal{C}_{\varphi,X}$ proceeds to check $w \models^? \evalnu{\varphi_j}{X}$. In order to ``switch'' to this new problem, $\mathcal{C}_{\varphi,X}$ needs to know $\varphi_j$, and so it maintains $\varphi_j$ it in its state. In other words, after $j$ steps $\mathcal{C}_{\varphi,X}$ is in state
$\left(\varphi_j, \aft(\evalnu{\varphi_i}{X}, w_{ij})\right)$, where $i \leq j$ is the number of steps after which $\mathcal{C}_{\varphi,X}$ switched to a new problem for the last time. If the second component of the state becomes $\false$, then the automaton uses the first component to determine which formula to check next. The accepting condition states that the transitions leading to a state of the form $(\psi, \false)$ must occur finitely often, which implies that eventually one of the checks  $w \models^? \evalnu{\varphi_j}{X}$ succeeds.

The formal description of $\mathcal{C}_{\varphi,X}$ is as follows:
\[\mathcal{C}_{\varphi,X} = \left(\Reach(\varphi) \times \evalnu{\Reach(\varphi)}{X}, \delta, (\varphi, \evalnu{\varphi}{X}), \fin(F) \right)\] where
\begin{itemize}
    \item $\evalnu{\Reach(\varphi)}{X} = \bigcup_{\psi \in \Reach(\varphi)} \Reach(\evalnu{\psi}{X})$
	\item $\delta((\xi, \zeta), \nu) = \begin{cases} (\aft(\xi, \nu), \aft(\evalnu{\xi}{X}, \nu)) & \text{if } \zeta \equiv_P \false \\ (\aft(\xi, \nu), \aft(\zeta, \nu)) & \text{otherwise.} \end{cases}$ 
	\item $F = \Reach(\varphi) \times \{\false\}$
\end{itemize}
\noindent Since $\Reach(\varphi) $ has at most size $2^{2^n}$, the number of states of $\mathcal{C}_{\varphi,X}$ is bounded by $\left(2^{2^n}\right)^2 = 2^{O(2^n)}$. 

\begin{example}
Let $\varphi = \G (a \U b \vee \F c)$, $X = \{a \U b\}$, and $\evalnu{\varphi}{X} = \G (a \W b)$. Below we show a fragment of  $\mathcal{C}_{\varphi,X}$, with $\alpha,\beta,\gamma$ as in Example  \ref{ex:GF}.

\begin{center}
	\begin{tikzpicture}[x=3cm,y=1.5cm,font=\footnotesize,initial text=,outer sep=2pt]
	\tikzstyle{acc}=[double]

	\node[state,initial] (1) at (0,0) {$\varphi, \evalnu{\varphi}{X}$};
	\node[state,acc] (2) at (2,0)  {$\varphi, \false$};
	\node[state] (3) at (1,0) {$\begin{array}{c}\varphi \wedge (a \U b \vee \F c), \\ \evalnu{\varphi}{X} \wedge a \W b\end{array}$};
	\node[] (4) at (0,-0.75) {};
	\node[] (5) at (1,-0.75) {};
	\node[] (6) at (2,-0.75) {};
	
	\path[->]  
	(1) edge[loop above] node[above]{$\beta$} (1)
	(1) edge[bend left] node[above]{$\{c\}$} (2)
	(1) edge[bend right=15] node[below]{$\{a\}$} (3)
	(2) edge node[below]{$\{a\}$} (3)
	(2) edge[loop above] node[above]{$\{c\}$} (2)
	(3) edge[bend right=15] node[below]{$\beta$} (1);
	
	\path[->, dashed] 
	(1) edge (4)
	(2) edge (6)
	(3) edge (5);
	
	\end{tikzpicture}
\end{center}

\noindent For $w = \{ c \} \{ c \} ( \{ a \} \{ b \})^\omega$ we have $X = \setGF_w$; the word is accepted. For $w' = \{c\}^\omega$ we have $X \neq \setGF_{w'}$, and the word is rejected. \qee
\end{example}

\paragraph{A DBA for $L_{X,Y}^2$.} We define a DBA recognizing
$L\left(\bigwedge_{\psi \in X} \G\F(\evalmu{\psi}{Y})\right)$. Observe that 
$\G\F (\evalmu{\psi}{Y}) \in \G\F(\ltlfum)$ for every $\psi \in X$, and that  $\evalmu{\psi}{Y}$ 
has at most $n$ subformulas. By Proposition \ref{prop:simpleaut}, $L(\G\F (\evalmu{\psi}{Y})$ is recognized by a DBA with at most $2^{2^{O(n)}}$ states. Recall that the intersection of the languages of $k$ DBAs with $s_1, \ldots, s_k$ states is recognized by a DBA with $k \cdot\prod_{j=1}^k s_j$ states.  Since $|X| \leq n$, the intersection of the DBAs for the formulas $\G\F (\evalmu{\psi}{Y})$ yields a DBA with at most $n \cdot \left(2^{2^{O(n)}}\right)^n = 2^{n2^{O(n)}} = 2^{2^{O(n)}}$ states.

\paragraph{A DCA for $L_{X,Y}^3(\varphi)$.} The DCA for
$L\left( \bigwedge_{\psi \in Y} \F\G(\evalnu{\psi}{X}\right)$ is obtained dually to the previous case, applying $\F\G (\evalnu{\psi}{X}) \in \F\G(\ltlgrw)$, and 
Proposition \ref{prop:simpleaut}.

\paragraph{A DRA for $L(\varphi)$.} By the Master Theorem we have:
\[L(\varphi) = \bigcup_{\substack{X \subseteq \sfmu(\varphi)\\Y \subseteq \sfnu(\varphi)}} L_{X,Y}^1 \cap L_{X,Y}^2 \cap L_{X,Y}^3\]
\noindent We first construct a DRA $A_{X,Y}$ for the intersection of $L_{X,Y}^i$, where $i=1,2,3$.
Let $A_{X,Y}^i$ be the DCA or DBA for $L_{X,Y}^i$.
The set  of states of $A_{X,Y}$ is the cartesian product of the sets of states of the $A_{X,Y}^i$, the transition function is as usual, and the accepting condition is \[\fin\left((S_1 \times Q_2 \times Q_3) \cup (Q_1 \times Q_2 \times S_3)\right) \wedge \inf(Q_1 \times S_2 \times Q_3)\]
\noindent where $Q_i$ is the set of states of $A_{X,Y}^i$, and $\fin(S_1)$, $\inf(S_2)$, $\fin(S_3)$ are the accepting conditions of $A_{X,Y}^1$, $A_{X,Y}^2$, and $A_{X,Y}^3$.

We construct a DRA $A_\varphi$ for $L(\varphi)$. Since $X \subseteq \sfmu(\varphi)$ and $Y \subseteq \sfnu(\varphi)$, there are at most $2^n$ pairs of sets $X, Y$.  Let $A_1, \ldots, A_k$ be an enumeration of the DRAs for these pairs, where $k \leq 2^n$, and let $Q_i$ and $\alpha_i~=~\fin(U_i) \wedge \inf(V_i)$ be the set of states and accepting condition of $A_i$, repectively. The set of states of $A_\varphi$ is $Q_1 \times \cdots \times Q_k$, the transition function is as usual, and the accepting condition is
\[\begin{array}{lll}
\bigvee_{i=1}^k &  \fin(Q_1 \times \cdots \times Q_{i-1} \times U_i \times Q_{i+1} \times \cdots \times Q_k)  & \wedge \\
& \inf(Q_1 \times \cdots \times Q_{i-1} \times V_i \times Q_{i+1} \times \cdots \times Q_k)
\end{array}\]
So $A_\varphi$ has $\left(2^{2^{O(n)}}\right)^{2^n} = 2^{2^{O(n)} \cdot 2^n} = 2^{2^{O(n)}}$
\noindent states and at most $2^n$ Rabin pairs.

%% file: Sections/construction-nba.tex
\section{Constructing NBAs for LTL Formulas}
\label{sec:ltl2nba}

Assume that $\varphi$ has length $n$. We use the Master Theorem to construct a NBA
for $L(\varphi)$ with $2^{O(n)}$ states.

We first describe how to construct NBAs for the LTL fragments of Section \ref{sec:fragments}. 
Let us start with some informal intuition. Consider the formula $\varphi = \G \X (a \vee b)$. In the DRA for $\varphi$ we find states for the formulas $\varphi$ and $\aft(\varphi, \emptyset)$ and a transition \[\varphi \trans{\emptyset} \aft(\varphi, \emptyset)\] where  $\aft(\varphi, \emptyset) \equiv_P  \varphi \wedge (a \vee  b)$.
\noindent The languages recognized from the states $\varphi$ and $\aft(\varphi, \emptyset)$ are precisely $L(\varphi)$ and $L(\aft(\varphi, \emptyset))$. The basic principle for the construction of the NBAs is to put $\aft(\varphi, \emptyset)$ in disjunctive normal form (DNF)
\[\varphi \wedge (a \vee  b) \equiv_P (\varphi \wedge a) \vee (\varphi \wedge b)\]
and instead of a single transition, have two transitions 
\[\varphi \trans{\emptyset} \varphi \wedge a \quad \text{ and } \quad \varphi \trans{\emptyset} \varphi \wedge b.\] In other words, the nondeterminism is used to guess which of the two disjuncts of the DNF is going to hold. Formally, we proceed as follows:

\begin{definition}
We define $\dnf(\varphi)$ as the set of clauses obtained by putting the propositional formula $\varphi$ in DNF, i.e., $\varphi \equiv_P \bigvee_{\psi \in \dnf(\varphi)} \psi$. Further let \[\DReach(\varphi) = \bigcup_{w \in {(2^{Ap})}^*} \daft(\psi, w) \] \noindent with $\daft(\psi, \epsilon) = \dnf(\psi)$, $\daft(\psi, \nu) = \dnf( \aft(\psi, \nu))$, and $\daft(\psi, \nu w) = \bigcup_{\psi' \in \daft(\psi, \nu)} \daft(\psi', w)$ for every formula $\psi$, letter $\nu$, and word $w$.
\end{definition}

\noindent Notice that  $\dnf(\false) = \emptyset$ and $\dnf(\true)= \{ \true \}$. Since the automata defined below have sets of states of the form $\DReach(\varphi)$, they have a state labeled by $\true$, but no state labeled by $\false$. 

The proof of the next proposition follows immediately from the definitions.
\begin{proposition}\label{prop:simplenbas}
Let $\varphi \in \ltlfum$. 
\begin{itemize}
\item The following NBA over the alphabet $2^{Ap}$ recognizes $L(\varphi)$:
\[\mathcal{A}_\mu^\varphi= (\, \DReach(\varphi), \daft, \dnf(\varphi), \inf(\true) \, )\]
\item The following NBA over the alphabet $2^{Ap}$ recognizes $L(\G\F\varphi)$:
\[\mathcal{A}_{\G\F\mu}^\varphi= (\, \DReach(\F\varphi), \daft_{\F \varphi},\{\F \varphi\}, \inf(\true) \,)\]
\[\daft_{\F \varphi}(\psi, \nu) = \begin{cases} \{\F \varphi\} & \text{if } \psi \equiv_P \true \\ \daft(\psi, \nu) & \text{otherwise.} \end{cases}\]
\end{itemize}

Let $\varphi \in \ltlgrw$. 
\begin{itemize}
\item The following NBA over the alphabet $2^{Ap}$ recognizes $L(\varphi)$:
\[\mathcal{A}_\nu^\varphi= (\, \DReach(\varphi), \daft, \dnf(\varphi), \inf(\DReach(\varphi)) \, )\]
\item The following NBA over the alphabet $2^{Ap}$ recognizes $L(\F\G\varphi)$:
\[\mathcal{A}_{\F\G\nu}^\varphi= (\, \DReach(\G\varphi) \cup \{\F\G\varphi\}, \daft_{\G \varphi}, \{\F \G \varphi\}, \inf(\DReach(\G\varphi)) \,)\]
\[\daft_{\G \varphi}(\psi, \nu) = \begin{cases} \{\F\G\varphi, \G \varphi\} & \text{if } \psi = \F\G\varphi \\ \daft(\psi, \nu) & \text{otherwise.} \end{cases}\]
\end{itemize}
\end{proposition}

Recall that the elements of $\Reach(\varphi)$ are positive boolean combinations of proper
subformulas of $\varphi$. It follows that the elements of $\DReach(\varphi)$ are \emph{conjunctions} of proper subformulas of $\varphi$. Since the number of proper subformulas is bounded by the length of the formula, we immediately obtain:

\begin{proposition}
If $\varphi$ has $n$ proper subformulas, then $\DReach(\varphi)$ has at most $2^n$ elements, and so all the NBAs of
Proposition \ref{prop:simplenbas}  have at most $2^{n+1}+1 = O(2^n)$ states.
\end{proposition}

\begin{example}\label{ex:GF:NBA}
Let $\varphi = a \wedge \X (b \vee \F c)$, the formula for which a DBA was given in Example \ref{ex:GF}. The NBA $\mathcal{A}_{\G\F\mu}^\varphi$ is shown below. The figure uses the abbreviations of Example \ref{ex:GF}.
	
\begin{center}
	\begin{tikzpicture}[x=3cm,y=1.5cm,font=\footnotesize,initial text=,outer sep=2pt]
	\tikzstyle{acc}=[double]

	\node[state,initial] (1) at (0,0) {$\F \varphi$};
	\node[state]         (2) at (1,0) {$\F c$};
	\node[state]         (3) at (1,-0.75) {$b$};
	\node[state,acc] 	 (4) at (2,0) {$\true$};
	
	\path[->] 
	(1) edge node[above]{$\alpha$} (2)
	(1) edge[loop above] node[above]{$2^{Ap}$} (1)
	(1) edge[bend right= 15] node[above]{$\alpha$} (3)
	(2) edge[loop above] node[above]{$2^{Ap}$} (2)
	(2) edge node[above]{$\gamma$} (4)
	(3) edge[bend right= 15] node[above]{$\beta$} (4)
	(4) edge[bend left=15, out=-45, in=225] node[above]{$2^{Ap}$} (1);
	\end{tikzpicture}
\end{center}
\noindent Compared to the DBA of Example \ref{ex:GF}, the NBA has a simpler structure, although in this case the same number of states. \qee
\end{example}

To define NBAs for arbitrary formulas we apply the Master Theorem. This is routine, and so we only sketch the constructions.

\paragraph{A NBA for $L_{X,Y}^1$.}
We define a NBA $\mathcal{C}_{\varphi,X}$ that accepts a word $w$ if{}f $w_{i} \models \evalnu{\aft(\varphi, w_{0i})}{X}$ for some suffix $w_i$ of $w$. Recall that $\evalnu{\aft(\varphi, w_{0i})}{X} \in \ltlgrw$ for every $i \geq 0$.
The automaton consists of two components with sets of states $Q_1$ and $Q_2$ given by 
\[Q_1 = \{ (\psi, 1) \mid \psi \in \DReach(\varphi) \} \quad  Q_2 = \{ (\evalnu{\psi}{X}, 2) \mid \psi \in \DReach(\varphi) \} \] \noindent 
Transitions either stay in the same component, or ``jump'' from the first component to the second. Transitions that stay in the same component are of the form $(\psi, i) \trans{\nu} (\psi', i)$ for $\psi' \in \daft(\psi, \nu)$ and $i=1,2$. ``Jumps'' are transitions of the form $(\psi, 1) \trans{\epsilon} (\evalnu{\psi}{X}, 2)$. Jumping amounts to nondeterministically guessing the suffix $w_i$ satisfying $\evalnu{\aft(\varphi, w_{0i})}{X}$. The accepting condition is $\inf(Q_2)$. Notice that the state $(\false, 2)$ does not have any successors.

\noindent Since $\DReach(\varphi)$ has at most $2^n$ states, $\mathcal{C}_{\varphi,X}$ has $2^{O(n)}$ states.

\paragraph{A NBA for $L_{X,Y}^2$.} As in the case of DRAs, we define a NBA recognizing 
$L\left(\bigwedge_{\psi \in X} \G\F(\evalmu{\psi}{Y})\right)$. To obtain an NBA with $2^{O(n)}$
states  we use a well-known trick.  Given a set $\{ \psi_1, \ldots, \psi_k \}$ of formulas,
we have 
\[\bigwedge_{i=1}^k \G\F \psi_i \equiv \G\F( \psi_1 \wedge \F (\psi_2 \wedge \F( \psi_3 \wedge \ldots \wedge \F (\psi_{k-1} \wedge \F \psi_k)\ldots))\]
The formula obtained after applying the trick belongs to $\G\F(\ltlfum)$ and has $O(n)$ $\mu$-subformulas. By Proposition \ref{prop:simplenbas}.2 we can construct a NBA for it with $2^{O(n)}$ states.

\paragraph{A NBA for $L_{X,Y}^3$.} In this case we apply
\[\bigwedge_{i=1}^k \F\G \psi_i \equiv \F\G \left( \bigwedge_{i=1}^k \psi_i \right)\]
and Proposition \ref{prop:simplenbas}.4, yielding an automaton with $2^{O(n)}$ states.

\paragraph{A NBA for $L(\varphi)$.} We proceed as in the case of DRAs, using the well-known operations for union and intersection of NBAs.  The NBA $A_\varphi$ is the union of at most $2^n$ NBAs $A_{X,Y}$, each of them 
with $2^{O(n)}$ states. The difference with the DRA case is that, given NBAs with $n_1, \ldots, n_k$ states accepting languages $L_1, \ldots, L_k$, we can construct a NBA for $\bigcup_{i=1}^k L_i$ with $\sum_{i=1}^k n_i$ states, instead of 
$\prod_{i=1}^k n_i$ states, as was the case for DRAs.  So $A_\varphi$ has $2^n \cdot 2^{O(n)} = 2^{O(n)}$ states.

\section{Constructing LDBAs for LTL Formulas}
\label{sec:ltl2ldba}

The translation of LTL into LDBA combines the translations into DRA and NBA. Recall that the states of an LDBA are partitioned into an initial component and a deterministic accepting component containing all accepting states. While in the definition of a LDBA the initial component can be nondeterministic, in our construction we can easily make it deterministic: Every accepting run has exactly one non-deterministic step.  This makes the LDBA usable for quantitative (and not only qualitative) probabilistic model checking, as described in \cite{DBLP:conf/cav/SickertEJK16}.

Lemma \ref{lem:delay} shows that checking property (1) of Theorem \ref{thm:logical:char} can be arbitrarily delayed, which allows us to slightly rephrase the Master Theorem as follows:

\begin{theorem}(Variant of the Master Theorem)
For every formula $\varphi$ and for every word $w$: \, $w \models \varphi$ if{}f there exists $X \subseteq \sfmu(\varphi)$, $Y \subseteq \sfnu(\varphi)$, and $i \geq 0$ satisfying
\[\begin{array}{ll}
(1')  & \quad w_i \models \evalnu{\aft(\varphi, w_{0i})}{X} \\
(2')  & \quad \forall \psi \in X. ~ w_i \models \G\F(\evalmu{\psi}{Y})  \\
(3')  & \quad \forall \psi \in Y. ~ w_i \models \G(\evalnu{\psi}{X})	
\end{array}\]
\end{theorem}
\begin{proof}
Clearly, the existence of an index $i$ satisfying (1'-3') implies that conditions (1-3) hold. For the other direction, assume conditions (1-3) hold. By Lemma \ref{lem:delay} the index $i$ of condition (1) 
can be chosen arbitrarily large. Since $w \models \bigwedge_{\psi \in X} \F\G(\evalnu{\psi}{X})$, we can choose $i$ so that it also satisfies $w_i \models \bigwedge_{\psi \in X} \G(\evalnu{\psi}{X})$. 
\end{proof}

The idea of the construction is to use the initial component to keep track of $\aft(\varphi, w_{0i})$---that is, after reading
a finite word $w_{0i}$ the initial component is in state $\aft(\varphi, w_{0i})$---and use the jump to the accepting component to guess sets $X$ and $Y$ and the stabilization point $i$. The jump leads to the initial state of the intersection of three DBAs, which are in charge of checking $(1')$, $(2')$, and $(3')$.

Recall that $\aft(\varphi, w_{0i}) \in \Reach(\varphi)$ for every word $w$ and every $i \geq 0$. 
For every $\psi \in \Reach(\varphi)$ and for each pair of sets $X, Y$  we construct a DBA $\mathcal{D}_{\psi,X,Y}$ recognizing the intersection of the languages of the formulas:
\begin{center}
$\evalnu{\psi}{X} \qquad \bigwedge_{\psi \in X}\G\F(\evalmu{\psi}{Y}) \qquad \bigwedge_{\psi \in Y}\G(\evalnu{\psi}{X})$	
\end{center}
\noindent These formulas belong to $\ltlgrw$, $\G\F(\ltlfum)$, and $\ltlgrw$, respectively, and so we can obtain DBAs for them following the recipes of  Proposition \ref{prop:simpleaut}. As argued before, each of these DBAs have $2^{2^{O(n)}}$ states,
and so we can also construct a DBA for their intersection with the same upper bound. Summarizing, we obtain:

\paragraph{Initial component.}
The component is $(\Reach(\varphi), \aft, \{\varphi\})$ and thus the component has at most $2^{2^n}$ states. Recall that this component does not have accepting states.

\paragraph{Accepting component.} The component is the disjoint union, for every $\psi \in \Reach(\varphi)$, 
$X \subseteq \sfmu(\varphi)$, and $Y \subseteq \sfnu(\varphi)$, of the DBA $\mathcal{D}_{\psi,X,Y}$. Since $\Reach(\varphi)$ has at most $2^{2^n}$ formulas and there are at most $2^n$ pairs $(X, Y)$, the component is the disjoint union of at most $2^{2^n} \cdot 2^n$ automata, each of then with $2^{2^{O(n)}}$ states. Thus in total $2^{2^{O(n)}}$ states.

\paragraph{A LDBA for $L(\varphi)$.} The LDBA is the disjoint union of the initial and accepting components.
The initial component is connected to the accepting component by $\epsilon$-transitions: For every formula $\psi \in \Reach(\varphi)$ and for every two sets $X, Y$, there is an $\epsilon$-transition from state $\psi$ of the initial component to the initial state of 
$\mathcal{D}_{\psi,X,Y}$. 

The LDBA has $2^{2^{O(n)}} + 2^{2^n} = 2^{2^{O(n)}}$ states. Recall that the lower bound for the blowup of a translation of LTL to LDBA is also doubly exponential (see e.g. \cite{DBLP:conf/cav/SickertEJK16}).

%% file: Sections/conclusion.tex
% !TEX root = ../main.tex

\section{Discussion}
\label{sec:discussion}

This paper builds upon our own work \cite{DBLP:conf/cav/KretinskyE12,DBLP:conf/atva/GaiserKE12,DBLP:conf/atva/KretinskyL13,DBLP:conf/cav/EsparzaK14,DBLP:journals/fmsd/EsparzaKS16,DBLP:conf/cav/SickertEJK16}. In particular, the 
notion of  \emph{stabilization point} of a word with respect to a formula, and the idea of using oracle information that is subsequently checked are already present there. 
The translations of LTL to LDBAs of \cite{DBLP:conf/tacas/KiniV15,DBLP:conf/tacas/Kini017} are based on similar ideas, also with resemblance to obligation sets of \cite{DBLP:conf/birthday/LiP0WHL13,DBLP:journals/fac/LiZZPVH18}.

The essential novelty of this paper with respect to the previous work is the
introduction of the symmetric mappings $\evalmu{\cdot}{\cdot}$ and $\evalnu{\cdot}{\cdot}$. Applying them to an arbitrary formula $\varphi$ yields a simpler formula, but \emph{not in the sense one might expect}. In particular,  $\evalmu{\varphi}{Y}$ may be \emph{stronger} than $\varphi$. For example, the information that, say, the formula  $a\W b$ does not hold infinitely often makes us check the \emph{stronger} formula $a \U b = \evalmu{(a \W b)}{\emptyset}$. However, exactly this point makes the ``$\mu$-$\nu$-alternation'' work: The formulas $\evalnu{\varphi}{X}$ and $\evalmu{\varphi}{Y}$ are only simpler in the sense of \emph{easier to translate}.  
This is the reason why operators $\W$ and $\M$ are present in the core syntax and the missing piece since the symmetric solutions \cite{DBLP:conf/cav/KretinskyE12,DBLP:conf/atva/KretinskyL13}, limited to fragments based on the simpler operators $\F$ and~$\G$.

The Master Theorem can be applied beyond what is described in this paper. In order to translate LTL into universal automata we only need to normalize formulas into conjunctive normal form. Furthermore one can obtain a double exponential translation into deterministic parity automata adapting the  approach described in \cite{DBLP:conf/tacas/EsparzaKRS17}.
Another intriguing question is whether our translation into NBA, which is very different from the ones described in the literature is of advantage in some application like runtime verification. 

The target automata classes used in practice typically use an acceptance condition defined on transitions, instead of states. Further, they use \emph{generalized} acceptance conditions, be it B\"uchi or Rabin. All our constructions can be restated effortlessly to yield automata with transition-based acceptance, and if generalized acceptance conditions are allowed then they become simpler and more succinct. The implementation used in our experiments actually uses these two features, which is described in the appendix of \cite{arXiv}.

To conclude, in our opinion this paper successfully finishes the journey started in \cite{DBLP:conf/cav/KretinskyE12}. Via a single theorem it provides an arguably elegant (unified, symmetric, syntax-independent, not overly complex) and efficient (asymptotically optimal and practically relevant) translation of LTL into your favourite $\omega$-automata.

\paragraph{Acknowledgments.} The authors want to thank Alexandre Duret-Lutz, Benedikt Seidl, the anonymous reviewers, and the participants of the 2nd Jerusalem Winter School in Computer Science and Engineering on \enquote{Formal Verification} for their helpful comments and remarks.

%% file: Appendix/proofs.tex
\section{Definitions and Proofs}

\subsection{Properties of $\aft$}

For convenience we restate the proof for the following existing result:

\lemAf*
\begin{proof}
(1) We show by induction on $\varphi$ that for a single letter $\nu \in 2^{Ap}$ the property $\nu w' \models \varphi \leftrightarrow w' \models \aft(\varphi, \nu)$ holds, where we just show two representative cases of the induction. The result for arbitrary $w$ is then proven by induction on the length of $w$. Let us now proceed with proving the single-letter case: \[\nu w' \models \varphi ~~ \textit{if{}f} ~~ w' \models \aft(\varphi, \nu)\]

\noindent Case $\varphi = a$. \[\nu w' \models a \leftrightarrow a \in \nu \leftrightarrow \aft(a, \nu) = \true \leftrightarrow \nu w \models \aft(a, \nu)\]

\noindent Case $\varphi = \psi_1 \U \psi_2$. 
\[\begin{array}{clr}
			    & \nu w' \models \varphi \\
\leftrightarrow & \nu w' \models \psi_2 \vee (\psi_1 \wedge \X \varphi) & \text{(LTL expansion)} \\
\leftrightarrow & w' \models \aft(\psi_2, \nu) \vee (\aft(\psi_1, \nu) \wedge \varphi) & \text{(LTL semantics and I.H.)} \\
\leftrightarrow & w' \models \aft(\varphi, \nu) \\
\end{array}\]

\noindent (2-3) Intuitively this holds, since $\aft$ does not create new temporal operators and maps only to Boolean combinations of existing temporal subformulas. Formally the proof proceeds by induction on $\varphi$ and then length of $w$. Thus each $\psi \in \Reach(\varphi)$ can be identified as a function over $n$ variables. Since there are most $2^{2^n}$ Boolean functions over the domain $\mathbb{B}^n$ the size of $\Reach(\varphi)$ is at most $2^{2^n}$. 
\end{proof}

\subsection{Definitions for $\sfmu(\varphi)$ and $\sfnu(\varphi)$}

\begin{definition}\label{def:sfmu}
Let $\varphi$ be a formula. The set $\sfmu(\varphi)$ is inductively defined as follows:
\begin{itemize}
\item If $\varphi = \true, \false, a, \neg a$, then $\sfmu(\varphi) = \emptyset$.
\item If $\varphi = \mathit{op}(\psi)$ for $\mathit{op} \in \{ \X, \G \}$ then $\sfmu(\varphi) = \sfmu(\psi)$.
\item If $\varphi = \mathit{op}(\psi_1, \psi_2)$ for $\mathit{op} \in \{ \wedge, \vee, \W, \R \}$ then \\ $\sfmu(\varphi) = \sfmu(\psi_1) \cup \sfmu(\psi_2)$.
\item If $\varphi = \F \psi$ then $\sfmu(\varphi) = \{\F \psi\} \cup \sfmu(\psi)$.
\item If $\varphi = \mathit{op}(\psi_1, \psi_2)$ for $\mathit{op} \in \{ \U, \M \}$ then \\ $\sfmu(\varphi) = \{\mathit{op}(\psi_1, \psi_2)\} \cup \sfmu(\psi_1) \cup \sfmu(\psi_2)$.
\end{itemize}
\end{definition}

\begin{definition}\label{def:sfnu}
Let $\varphi$ be a formula. The set $\sfnu(\varphi)$ is inductively defined as follows:
\begin{itemize}
\item If $\varphi = \true, \false, a, \neg a$, then $\sfnu(\varphi) = \emptyset$.
\item If $\varphi = \mathit{op}(\psi)$ for $\mathit{op} \in \{ \X, \F \}$ then $\sfnu(\varphi) = \sfnu(\psi)$.
\item If $\varphi = \mathit{op}(\psi_1, \psi_2)$ for $\mathit{op} \in \{ \wedge, \vee, \U, \M \}$ then \\ \quad $\sfnu(\varphi) = \sfnu(\psi_1) \cup \sfmu(\psi_2)$.
\item If $\varphi = \G \psi$ then $\sfnu(\varphi) = \{\G \psi\} \cup \sfnu(\psi)$.
\item If $\varphi = \mathit{op}(\psi_1, \psi_2)$ for $\mathit{op} \in \{ \W, \R \}$ then \\ \quad $\sfnu(\varphi) = \{\mathit{op}(\psi_1, \psi_2)\} \cup \sfnu(\psi_1) \cup \sfnu(\psi_2)$.
\end{itemize}
\end{definition}

\subsection{Properties of $\evalnu{\varphi}{X}$ and $\evalmu{\varphi}{Y}$}

\lemEvalnuEvalmu*
\begin{proof}
\noindent (a1) Assume $~\setF_w \subseteq X$. Then $\setF_{w_i} \subseteq X$ for all $i \geq 0$. We prove the following stronger statement via structural induction on $\varphi$:
\[\forall i. ~ (\; (w_i \models \varphi) \rightarrow (w_i \models \evalnu{\varphi}{X}) \; )\]

\noindent Case $\varphi = \psi_1 \U \psi_2$: Let $i \geq 0$ arbitrary and assume $w_i \models \psi_1 \U \psi_2$. Then $\psi_1 \U \psi_2 \in \setF_{w_i}$ and so $\psi_1 \U \psi_2 \in X$. We prove $w_i \models \evalnu{(\psi_1 \U \psi_2)}{X}$: 
\[
\arraycolsep=1.8pt
\begin{array}{clr}
 	        & w_i \models \psi_1 \U \psi_2 \\
\implies  & w_i \models \psi_1 \W \psi_2 \\ 
\implies  & \forall j. ~ w_{i+j} \models \psi_1 \vee \exists k \leq j.~ w_{i+k} \models \psi_2 \\
\implies  & \forall j. ~ w_{i+j} \models \evalnu{\psi_1}{X} \vee \exists k \leq j.~ w_{i+k} \models \evalnu{\psi_2}{X} & \text{(I.H.)} \\
\implies  & w_i \models (\evalnu{\psi_1}{X}) \W (\evalnu{\psi_2}{X}) \\
\implies  & w_i \models \evalnu{(\psi_1 \U \psi_2)}{X} & \text{(Def. $\evalnu{\varphi}{X}$)}
\end{array}
\]

\noindent Case $\varphi = \psi_1 \vee \psi_2$: Let $i \geq 0$ arbitrary and assume $w_i \models \psi_1 \vee \psi_2$. We have:
\[\begin{array}{clr}
	        & w_i \models \psi_1 \vee \psi_2  \\
\implies \; & (w_i \models \psi_1) \vee (w_i \models \psi_2) \\
\implies \; & (w_i \models \evalnu{\psi_1}{X}) \vee (w_i \models \evalnu{\psi_2}{X}) &  \text{(I.H.)} \\
\implies \; & w_i \models \evalnu{(\psi_1 \vee \psi_2)}{X} & \text{(Def. $\evalnu{\varphi}{X}$)}
\end{array}\]

\medskip

\noindent (a2) Assume $X \subseteq \setGF_w$. Then $X \subseteq \setGF_{w_i}$ for all $i \geq 0$. We prove the following stronger statement via structural induction on $\varphi$:
\[\forall i. ~ (\; (w_i \models \evalnu{\varphi}{X}) \rightarrow (w_i \models \varphi) \; )\]

\noindent Case $\varphi = \psi_1 \U \psi_2$: If $\varphi \notin X$, then by definition $\evalnu{\varphi}{X} = \false$. So $w_i \not \models \evalnu{\varphi}{X} = \false$ for all $i$ and thus the implication $(w_i \models \evalnu{\varphi}{X}) \rightarrow (w_i \models \varphi)$ holds for every $i \geq 0$. Assume now $\varphi \in X$. Since $X \subseteq \setGF_w$ we have $w_i \models \G \F \varphi$ and so in particular $w_i \models \F \psi_2$. To prove the implication assume $w_i \models \evalnu{(\psi_1 \U \psi_2)}{X}$ for an arbitrary fixed $i$. We show $w_i \models \psi_1 \U \psi_2$: 
\[
\arraycolsep=1.8pt
\begin{array}{clr}
	        & w_i \models \evalnu{(\psi_1 \U \psi_2)}{X} \\
\implies \; & w_i \models (\evalnu{\psi_1}{X}) \W (\evalnu{\psi_2}{X}) & \text{(Def. $\evalnu{\varphi}{X}$)}\\
\implies \; & \forall j. ~ w_{i+j} \models \evalnu{\psi_1}{X} \vee \exists k \leq j.~ w_{i+k} \models \evalnu{\psi_2}{X} \\
\implies \; & \forall j. ~ w_{i+j} \models \psi_1 \vee \exists k \leq j.~ w_{i+k} \models \psi_2 & \text{(I.H.)} \\
\implies \; & w_i \models \psi_1 \W \psi_2 \\
\implies \; & w_i \models \psi_1 \U \psi_2 & (w_i \models \F \psi_2)
\end{array}\]

\noindent Case $\varphi = \psi_1 \vee \psi_2$: Let $i \geq 0$ arbitrary and assume $w_i \models \psi_1 \vee \psi_2$. We have:
\[
\begin{array}{clr}
	        & w_i \models \evalnu{(\psi_1 \vee \psi_2)}{X} \\
\implies \; & (w_i \models \evalnu{\psi_1}{X}) \vee (w_i \models \evalnu{\psi_2}{X}) & \text{(Def. $\evalnu{\varphi}{X}$)}\\
\implies \; & (w_i \models \psi_1) \vee (w_i \models \psi_2) & \text{(I.H.)} \\
\implies \; & w_i \models \psi_1 \vee \psi_2
\end{array}\]

\noindent (b1) Assume $~\setFG_w \subseteq Y$. Then $\setFG_{w_i} \subseteq Y$ for all $i$. We prove the following stronger statement via structural induction on $\varphi$:
\[\forall i. ~ (\; (w_i \models \varphi) \rightarrow (w_i \models \evalmu{\varphi}{Y})\;)\]

\noindent Case $\varphi = \psi_1 \W \psi_2$: Let $i \geq 0$ arbitrary and assume $w_i \models \varphi$. If $\varphi \in Y$ then  $\evalmu{\varphi}{Y} = \true$ and so $w_i \models \evalmu{\varphi}{Y} $ trivially holds. Assume now $\varphi \notin Y$. Since $\setFG_{w_i} \subseteq Y$ we have $w_i \not \models \F\G \varphi$ and so in particular $w_i \not \models \G\psi_1$. We prove $w_i \models \evalmu{(\psi_1 \W \psi_2)}{Y}$: 
\[
\arraycolsep=1.8pt
\begin{array}{clr}
 	        & w_i \models \psi_1 \W \psi_2 \\
\implies \; & w_i \models \psi_1 \U \psi_2 & (w_i \not \models \G \psi_1) \\ 
\implies \; & \exists j. ~ w_{i+j} \models \psi_2 \wedge \forall k < j.~ w_{i+k} \models \psi_1 \\
\implies \; & \exists j. ~ w_{i+j} \models \evalmu{\psi_2}{Y} \wedge \forall k < j.~ w_{i+k} \models \evalmu{\psi_1}{Y} & \text{(I.H.)} \\
\implies \; & w_i \models (\evalmu{\psi_1}{Y}) \U (\evalmu{\psi_2}{Y}) \\
\implies \; & w_i \models \evalmu{(\psi_1 \W \psi_2)}{Y} & \text{(Def. $\evalmu{\varphi}{Y}$)}
\end{array}\]

\noindent Case $\varphi = \psi_1 \vee \psi_2$: Let $i \geq 0$ arbitrary and assume $w_i \models \psi_1 \vee \psi_2$. We have:
\[\begin{array}{clr}
	        & w_i \models \psi_1 \vee \psi_2  \\
\implies \; &  (w_i \models \psi_1) \vee (w_i \models \psi_2) \\
\implies \; & (w_i \models \evalmu{\psi_1}{Y}) \vee (w_i \models \evalmu{\psi_2}{Y}) & \qquad \text{(I.H.)} \\
\implies \; & w_i \models \evalmu{(\psi_1 \vee \psi_2)}{Y} & \text{(Def. $\evalmu{\varphi}{Y}$)}\\
\end{array}\]

\noindent (b2) Assume $Y \subseteq \setG_w$. Then $Y \subseteq \setG_{w_i}$ for all $i$. We prove the following stronger statement via structural induction on $\varphi$:
\[\forall i. ~ (\;(w_i \models \evalmu{\varphi}{Y}) \rightarrow (w_i \models \varphi)\;)\]

\noindent Case $\varphi = \psi_1 \W \psi_2$: If $\varphi \in Y$, then since $Y \subseteq \setG_w$ we have $w_i \models \G \varphi$ and so $w_i \models \varphi$. Assume now that $\varphi \notin Y$ and $w_i \models \evalmu{(\psi_1 \W \psi_2)}{Y}$ for an arbitrary fixed $i$. We prove $w_i \models \psi_1 \W \psi_2$: 
\[
\arraycolsep=1.8pt
\begin{array}{clr}
	        & w_i \models \evalmu{(\psi_1 \W \psi_2)}{Y} \\
\implies \; & w_i \models (\evalmu{\psi_1}{Y}) \U (\evalmu{\psi_2}{Y}) & \text{(Def. $\evalmu{\varphi}{Y}$)}\\
\implies \; & \exists j. ~ w_{i+j} \models \evalmu{\psi_2}{Y} \wedge \forall k < j.~ w_{i+k} \models \evalmu{\psi_1}{Y} \\
\implies \; & \exists j. ~ w_{i+j} \models \psi_2 \wedge \forall k < j.~ w_{i+k} \models \psi_1 & \text{(I.H.)} \\
\implies \; & w_i \models \psi_1 \U \psi_2 \\
\implies \; & w_i \models \psi_1 \W \psi_2
\end{array}\]

\noindent Case $\varphi = \psi_1 \vee \psi_2$: We derive in a straightforward manner for an arbitrary and fixed $i$: 
\[\begin{array}{clr}
	        & w_i \models \evalmu{(\psi_1 \vee \psi_2)}{Y} \\
\implies \; & (w_i \models \evalmu{\psi_1}{Y}) \vee (w_i \models \evalmu{\psi_2}{Y}) & \text{(Def. $\evalmu{\varphi}{Y}$)}\\
\implies \; & (w_i \models \psi_1) \vee (w_i \models \psi_2) & \text{(I.H.)} \\
\implies \; & w_i \models \psi_1 \vee \psi_2 \\
\end{array}\]
\end{proof}